\documentclass[a4paper]{paper}
\usepackage[utf8]{inputenc}
\usepackage{fullpage}
\usepackage{color}
\usepackage{bbm}
\usepackage{lscape}
\usepackage{rotating}
\usepackage{amsmath}
\usepackage{amssymb}
\usepackage{amsfonts}
\usepackage{amsthm}
\usepackage{graphicx}
\usepackage{natbib}
\usepackage{tikz}

\DeclareMathAlphabet{\mathpzc}{OT1}{pzc}{m}{it}
\def\pig{\boldsymbol{\pi}} 
\def\Ab{{\bf A}}
\def\B{{\rm B}}
\def\Jb{{\bf J}}
\def\rb{{\bf r}}
\def\pb{{\bf p}}
\def\fb{{\bf f}}
\def\E{{\rm E}}
\def\I{{\rm I}}
\def\Var{{\rm var}}
\def\Pr{{\rm Pr}}
\def\R{\mathbb R}
\def\N{\mathbb N}
\def\1{\mathbf{1}}
\def\ybar{\overline{Y}}

\title{Sampling Designs on Finite Populations\\ with  Spreading Control Parameters}

\newtheorem{lemma}{Lemma}[section]
\newtheorem{definition}{Definition}[section]

\newtheorem{proposition}{Proposition}[section]
\newtheorem{remark}{Remark}[section]
\newtheorem{corollaire}{Corollary}[section]

\newtheorem{exam}{Example}[section]

\begin{document}


\author{Yves Tillé\thanks{
Corresponding author. E-mail: yves.tille@unine.ch.}, \emph{\large University of Neuch\^atel} \and Lionel Qualité, \emph{\large Swiss Federal Office of Statistics and  University of Neuch\^atel} \and Matthieu Wilhelm,  \emph{\large University of Neuch\^atel}}


\maketitle
\begin{abstract}
We present new sampling methods in finite population that allow one to control the joint inclusion probabilities of units and especially the spreading of sampled units in the population. They are based on the use of renewal chains and multivariate discrete distributions to generate the difference of population ranks between successive selected units. With a Bernoulli sampling design, these differences follow a geometric distribution, and with a simple random sampling design they follow a negative hypergeometric distribution. We propose to use other distributions and introduce a large class of sampling designs with and without fixed sample size. The choice of the rank-difference distribution allows us to control units joint inclusion probabilities with a relatively simple method and closed form formula. Joint inclusion probabilities of neighboring units can be chosen to be larger, or smaller, compared to those of Bernoulli or simple random sampling, thus allowing more or less spread  of the sample in the population. This can be useful when neighboring units have similar characteristics or, on the contrary, are very different. A set of simulations illustrates the qualities of this method.
\end{abstract}

\section{Introduction}

In this paper, we propose sampling methods for fixed and random sample sizes. We more particularly focus on the spacings that are the difference of population ranks between two successive selected units. We propose a large set of new methods that allows one to control the spacings and thus the joint inclusion probabilities of population units in the sample. These methods are useful in that they allow one to make less (or more) likely the selection of neighboring units. Indeed, when the variable of interest takes similar values on neighboring units, spreading the sample improves estimation because the selection of similar units is avoided.

A sampling design is a probability distribution on all the finite subsets of a population. It can be implemented by means of sampling algorithms. Several different sampling algorithms can implement the same sampling designs. Examples are given in \citet{til:06} where a large number of algorithms is given for designs like Simple Random Sampling (SRS) with and without replacement or maximum entropy sampling designs. Algorithms such that the decision of selecting or not a unit into the sample is taken for each population unit successively according to the order of the population sampling frame are called ``sequential'' or ``one-pass'' algorithms. These algorithms are particularly useful when the population list is dynamic, like on a production chain or in real time sampling applications.

Systematic sampling is one of the most common sampling designs. It has been studied among others by \citet{mad:mad:44}, \citet{coc:46}, \citet{mad:49}, \citet{bel:rao:75}, \citet{iac:82}, \citet{iac:83}, \citet{mur:rao:88}, \citet{bel:88b}, \citet{Bell:Sutr:vari:1988}, and \citet{pea:qua:til:07}. One advantage of systematic sampling is that it spreads the sample very well over the population, thus allowing one to get precise estimators for totals and averages in the case of ``auto-correlated'' interest variables. Indeed, it can be shown to be an optimal design in this case under some conditions \citep{Bondesson86}. However, it presents the important drawback that lots of unit couples have null joint inclusion probabilities. This makes impossible an unbiased estimation of the variance.

This drawback has led to a quest for other sampling designs that would retain good estimation properties. \citet{dev:98a} proposed the Deville-systematic method,  also called ordered pivotal method by \citet{chauvet2012characterization} (see also \citet[][pp. 128-130]{til:06}). \citet{til:96c} proposed a moving stratification algorithm that avoids the selection of neighboring units. \citet{bon:tho:08} and \citet{Grafstrom2010982} also proposed a method that allows one to control joint-inclusion probabilities. Recently, \cite{2015arXiv151006618L} proposed using determinantal point processes that are known for their repulsiveness property (see for example \citet[][p.138]{Daley02}. This last method necessitates to work with a huge matrix.

We advocate the use of point processes with simple specifications, motivated by usual sampling designs: the systematic design has deterministic spacings between selected units, the Bernoulli sampling design (see for example \citet[][pp.43--44]{til:06}) has geometrically distributed spacings, and circular spacings of the simple random sampling design follow a negative hypergeometric distribution (see \citet{vit:84,vit:85,vit:87}). In this paper, we will use other distributions to tune the joint selection probability of neighboring units. For each of these methods, we are able to compute positive joint inclusion probabilities and unbiased variance estimators. Special attention to edge effects must be given to ensure correct first-order inclusion probabilities. Part of these sampling designs, with independent and identically distributed spacings, were introduced by \citet{Bondesson86}. 

The paper is organized as follows: Section~\ref{rappels} is devoted to the main definitions of survey sampling theory. Sections~\ref{renewal} and~\ref{spreading} present renewal chain sampling designs for random size samples. In Sections~\ref{circular} and~\ref{spreading2}, we discuss fixed size sampling obtained through the generation of circular spacings with multivariate discrete distributions. Simulation results are given in Section~\ref{simulations}. The paper ends with our conclusions in Section~\ref{conclusion}.

\section{Sampling from a finite population}\label{rappels}

Consider the finite population of $N$ units, $U=\left\{1,\dots,N\right\}$. A sample without replacement of $U$ is a subset $s\subset U$. A sampling design $P(.)$ is a probability distribution on samples,
$$
P(s)\geq 0,\ s\subset U, \mbox{ such that }\sum_{s \subset U}P(s)=1.
$$
Let $S$ denote the random sample, so that $\Pr(S=s)=P(s)$. The sample size $n=\# S$ can be random or not. The inclusion probability of unit $k$ is its probability of being selected into a sample
$$
\pi_{k}=\Pr(k \in S)=\sum_{s \ni k } P(s).
$$
The joint inclusion probability of units $k$ and $\ell$ is their probability of being selected together into a sample
$$
\pi_{k\ell}=\pi_{\ell k}=\Pr(k \mbox{ and } \ell\in S)=\sum_{s \ni {k,\ell}} P(s).
$$

Let $Y$ be a variable of interest and let $y_k$ be the value of $Y$ associated to unit $k$ of the population. The \cite{hor:tho:52} estimator
is defined by
$$
\widehat{Y} =  \sum_{k\in S}\frac{y_{k}}{\pi_{k}}.
$$
It is an unbiased estimator of the population total
$$
t_{Y} = \sum_{k\in U} y_{k},
$$
provided that $\pi_{k}>0,$ $k\in U$.
Let
$$
\Delta_{k\ell}=
\left\{
\begin{array}{ll}
\pi_{k\ell}-\pi_{k}\pi_{\ell} & \mbox{ if }k\neq \ell, \\
\pi_{k}(1-\pi_{k}) & \mbox{ if }k= \ell.
\end{array}
\right.
$$
The variance of the HT-estimator is 
$$
\Var\left(\widehat{Y}\right) =  \sum_{k\in U}\sum_{\ell \in U} \frac{y_{k} y_{\ell}}{\pi_k\pi_{\ell}}\Delta_{k\ell}.
$$
If the sampling design has a fixed size, the variance can also be written as (see \citet{sen:53,yat:gru:53})
$$
\Var\left(\widehat{Y}\right) = -\frac{1}{2} \sum_{k\in U}\sum_{\substack{\ell\in U \\ \ell \neq k}}
\left(\frac{y_{k}}{\pi_{k}} - \frac{y_{\ell}}{\pi_{\ell}} \right)^2 \Delta_{k\ell}.
$$
Estimators can be derived from these two expressions. For the general case, the \cite{hor:tho:52} variance estimator is given by:
\begin{equation}\label{vht}
\widehat{\Var}_{HT}\left(\widehat{Y}\right) = \sum_{k\in S}\sum_{\ell \in S} \frac{y_{k} y_{\ell}}{\pi_{k}\pi_{\ell}}\frac{\Delta_{k\ell}}{\pi_{k\ell}},
\end{equation}
where $\pi_{kk}=\pi_{k}$. When the sample size is fixed, the \citet{sen:53} and \citet{yat:gru:53} variance estimator is given by
\begin{equation}\label{vsyg}
\widehat{\Var}_{SYG}\left(\widehat{Y}\right) = -\frac{1}{2} \sum_{k\in S}\sum_{\substack{\ell\in S\\ \ell \neq k}}
\left(\frac{y_{k}}{\pi_{k}} - \frac{y_{\ell}}{\pi_{\ell}} \right)^{2} \frac{\Delta_{k\ell}}{\pi_{k\ell}}.
\end{equation}
These estimators are unbiased provided that $\pi_{k\ell}>0$, $k \neq \ell \in U$. Estimator~(\ref{vsyg}) is non-negative when $\Delta_{k\ell}\leq 0$, $k \neq \ell$ (Sen-Yates-Grundy conditions).

\section{Renewal chain sampling designs}\label{renewal}
The idea of selecting samples through the use of renewal processes is not new. It can be traced back at least to \citet{Bondesson86} (see also \citet{Meister04}). We give a different presentation in this section in that we focus on the parametrization of the distribution of spacings between selected units whereas \citet{Bondesson86} and \citet{Meister04} focus on the parametrization of the so-called \emph{renewal sequence}, the conditional inclusion probabilities given the past. Their aim was to provide solutions for real time sampling, and the proposed methods are intrinsically sequential, allowing one to spread the sample by introducing a negative correlation between the sample inclusion indicators. \citet{bon:tho:08} generalize this idea using a splitting method (see \citet[][]{dev:til:98}) that allows use of a unequal probability sampling designs for real time sampling.

\subsection{Definition}

In this section, we present a family of sampling algorithms that are parametrized by a discrete probability distribution. By a careful choice of the generating distribution, we obtain sampling designs with desirable properties.
Consider a sequence $J_{1},\dots,J_{N}$ of independently and identically distributed (i.i.d.) random variables in $\N^{*}=\{1,2,3,\dots\}$. The partial sums $S_{j}=\sum_{i=1}^{j}J_{i}$, $j\ge 1$, form a discrete process that is called a simple renewal chain (see for example \citet[][]{fel:71}, \citet[][p.18]{bar:lim:08}), by analogy with renewal processes (see \citet[][]{cox1962renewal, Daley02, Mitov14}). Using these $J_{i}$'s as spacings (jumps) between successive units selected into the sample, we obtain the family of sampling designs of Definition~\ref{def1}.
\begin{definition}\label{def1}
A sampling design is said to be a (simple) renewal chain sampling design if its random sample can be written
$$
\tilde{S}=\{1,\dots,N\}\bigcap \left\{\sum_{i=1}^{j}J_{i},\ 1\leq j\leq N\right\},
$$
where $J_{1},\dots,J_{N}$ are i.i.d. random variables in $\N^{*}$.
\end{definition}
The first-order inclusion probability of a renewal chain design can be obtained from the common distribution $f(\cdot)$ of the $J_{i}$'s:
\begin{equation}\label{eqpi}
\pi_{k}=\Pr(k \in \tilde{S})=\sum_{j=1}^{k}f^{j*}(k),
\end{equation}
where $f^{j*}(\cdot)$ is the distribution of the sum of $j$ i.i.d. variables with distribution $f(\cdot)$. Indeed, unit $k$ is selected if $J_{1}=k$, or $J_{1}+J_{2}=k$, or~$\cdots$, or $J_{1}+\cdots+J_{k}=k$. These events are non-overlapping thanks to the $J_{i}$'s being positive. We obtain that:
$$
\pi_{k} = \sum_{j=1}^{k} \Pr\left(\sum_{i=1}^{j} J_{i}=k \right),
$$
which is exactly Equation~(\ref{eqpi}). It is a well-known property of renewal process theory given, for example, in \citet[][p. 21]{bar:lim:08}, \citet[][p. 53]{cox1962renewal} or in \citet[][pp. 44-47]{Mitov14}.

Even with i.i.d. spacings, a simple renewal chain sampling design usually has unequal first order inclusion probabilities, as we can see in Example~\ref{example1}.
\begin{exam}\label{example1}
Let $J_{i}$, $i\in\N^{*}$ be a sequence of i.i.d. variables such that $\Pr(J_{i}=1)=1/2$ and $\Pr(J_{i}=2)=1/2$. Then,
\begin{eqnarray*}
\pi_{1} & = & \Pr(J_{1}=1)=1/2,\\
\pi_{2} & = & \Pr(J_{1}=2)+\Pr(J_{1} +J_{2}=2)=1/2+1/4 = 3/4,\\
\pi_{3} & = & \Pr(J_{1}+J_{2}=3)+\Pr(J_{1}+J_{2}+J_{3}=3) =1/2+1/8=5/8,\\
\pi_{4} & = & \Pr(J_{1}+J_{2}=4)+\Pr(J_{1}+J_{2}+J_{3}=4) +\Pr(J_{1}+J_{2}+J_{3}+J_{4}=4) 
=11/16,\\
&\vdots&
\end{eqnarray*}
\end{exam}

\subsection{Equilibrium renewal chains}

A delayed renewal chain is a discrete process $(S_{j})_{j \in \N}$ with $S_{j}=\tilde{J}_{0}+\sum_{i=1}^{j}J_{i}$,  where the $J_{i}$'s, $i\geq 1$ are i.i.d. random variables taking values in $\N^{*}$ and $\tilde{J}_{0}$ is an independent random variable taking values in $\N$ (see e.g. \citet[][p. 31]{bar:lim:08}). Of particular interest is the delayed renewal chain obtained when the distribution of $\tilde{J}_{0}$ is obtained from the distribution of $J_{1}$ using
\begin{equation}\label{J0}
\Pr(\tilde{J}_{0}=k)=\frac{\Pr(J_{1}\geq k+1)}{\E(J_{1})},\ k\in\N,
\end{equation}
provided that $\E(J_{1})$ exists. The distribution of $\tilde{J}_{0}$ is called the stationary or equilibrium distribution of the renewal chain and the resulting delayed renewal chain is called an equilibrium renewal chain. As written by \citet[][Proposition 2.2]{bar:lim:08}, this choice of the initial distribution $\tilde{J}_{0}$ of the delayed renewal chain is the only one where all $k \in \N$ have the same probability of being in the sample path. Proposition~\ref{fondamental_disc} is a general result of renewal process theory (see for example \citet[][p. 46]{Mitov14}) that we applied to the discrete case. We propose a direct proof of Proposition~\ref{fondamental_disc} in Appendix.
\begin{proposition}\label{fondamental_disc}
If $f(\cdot)$ is a probability distribution on $\N^{*}$ with cumulative distribution function $F(\cdot)$, expectation $\mu$, and if $f_{0}(\cdot)$ is defined by
$
f_{0}(k)=f(\{k+1,\dots\})/\mu, \ k\in\mathbb{N},
$
then $f_{0}(\cdot)$ is a probability distribution and
\begin{equation}\label{flat}
f_{0}(k)+\sum_{t=1}^{k} f_{0}(k-t)\sum_{j=1}^{t}f^{j*}(t) = \frac{1}{\mu},\mbox{ for all } k\geq 1.
\end{equation}
\end{proposition}
\begin{corollaire}\label{equpik}
Let $S_{j}$, $j \in \N$ be a delayed renewal chain with $\E(J_{1})=\mu$ and $\tilde{J}_{0}$ have the distribution of~(\ref{J0}). For all $k\in \N$, if $\pi_{k}$ is the probability that $k$ is in the sample path, then
\begin{equation}\label{pirenew}
\pi_{k}=f_{0}(k)+\sum_{t=1}^{k} f_{0}(k-t)\sum_{j=1}^{t}f^{j*}(t),
\end{equation}
and is equal to $1/\mu$.
\end{corollaire}
\begin{proof}
The event $\left\{\exists i\in \N\mbox{ such that }S_{i}=k\right\}$ can be decomposed as
$$
\left\{\exists i\in \N\mbox{ such that }S_{i}=k\right\}=\bigcup_{t=0}^{k}\left\{\tilde{J}_{0}=k-t\right\}\bigcap\bigcup_{j=1}^{t}\left\{\sum_{i=1}^{j}J_{i}=t\right\},
$$
where all the events in the union are non-overlapping. It follows that
$$
\pi_{k}=f_{0}(k)+\sum_{t=1}^{k} f_{0}(k-t)\sum_{j=1}^{t}f^{j*}(t)= \frac{1}{\mu},
$$
by Proposition~\ref{fondamental_disc}.
\end{proof}
\begin{definition}\label{forward}
Let $X$ be a random variable with values in $\N$ and finite expectation. A random variable $X_{F}$ is called a forward transform of $X$ if its distribution is given by
$$
\Pr(X_{F}=k) = \frac{\Pr(X\ge k)}{\E(X+1)},\ k\in  \N.
$$
\end{definition}
\begin{remark}\label{remarque}
Moments of $X_{F}$ can be derived from those of $X$ using the property, proven in the Appendix, that if $X$ is a random variable on $\N$ with finite moment of order $m+1$, $\E(X^{m+1})$, $m\geq 0$, then its forward transform $X_{F}$ has a finite moment of order $m$ and
$$
\E(X_{F}^{m})=\frac{\E[F_{m}(X)]}{\E(X+1)},
$$
where $F_{m}(x)$ is the Faulhaber polynomial integer function of degree $m+1$: $F_{m}(x)=\sum_{k=0}^{x}k^{m}$.
\end{remark}
The equilibrium distribution $\tilde{J}_{0}$ is the forward transform of the distribution of $J_{1}-1$, according to Definition~\ref{forward}. Spacing distributions considered in Section~\ref{spreading} are defined as shifted variables $J_{1}=1+X$ where $X$ follows a classical probability distribution on $\N$. The reader can find in Table~\ref{discrete} a collection of distributions that are used in Sections~\ref{spreading} and~\ref{spreading2}, as well as their forward transforms. 

\subsection{Equilibrium renewal chain sampling designs}

By taking the intersection of the sample path of an equilibrium renewal chain with the population  $U=\left\{1,\dots,N\right\}$, one obtains a random sampling design. Corollary~\ref{equpik} ensures that all units of the population have the same inclusion probability. The distribution of the first selected unit index $X_{1}$ satisfies Equation~(\ref{nj0}).
\begin{equation}\label{nj0}
\Pr(X_{1}=k)=\Pr(\tilde{J}_{0}=k)+\Pr(\tilde{J}_{0}=0)\Pr(J_{1}=k)=\frac{\Pr(J_{1}\geq k)}{\E(J_{1})},\ k\in U.
\end{equation}
By definition, the following sampled units are obtained by adding independent variables distributed like $J_{1}$. 
\begin{definition}\label{def3}
An equilibrium renewal chain sampling design is the distribution of a random sample $S$ with
$$
S=\left\{1,\dots,N\right\}\bigcap \left\{\sum_{i=0}^{j}J_{i},\ 0\leq j\leq N-1\right\},
$$
where $J_{1},\dots,J_{N-1}$ are i.i.d random variables in $\N^{*}$ with finite expectation, and $J_{0}$ is an independent variable with distribution given by~(\ref{nj0}), $\Pr(J_{0}=k)=\Pr(J_{1}\geq k)/\E(J_{1})$, $k\in U$.
\end{definition}
For $J_{1}=1+X$, the random variable $J_{0}$ of~(\ref{nj0}) has the same distribution as $1+X_{F}$ where $X_{F}$ is a forward transform of $X$.

The equilibrium renewal chain design that corresponds to the renewal distribution of Example~\ref{example1} is given in Example~\ref{example2}. Its first order inclusion probabilities are equal.
\begin{exam}
\label{example2}
Consider the sequence $J_{i}$, $i\in\N^{*}$ of Example~\ref{example1}, and define $J_{0}$ to be independent of the $J_{i}$'s, with $P(J_{0}=1)=2/3$ and $P(J_{0}=2)=1/3$ according to~(\ref{nj0}). The new inclusion probabilities $\tilde{\pi}_{i}$ of this equilibrium renewal sampling design are related to those of Example~\ref{example1} by:
\begin{eqnarray*}
\tilde{\pi}_{1}&=\Pr(J_{0}=1)\phantom{\pi_{2}+\Pr(J_{0}=2)\pi_{1}} &=2/3,\\
\tilde{\pi}_{2}&=\Pr(J_{0}=1)\pi_{1}+\Pr(J_{0}=2)\phantom{\pi_{1}}&=1/3+1/3=2/3,\\
\tilde{\pi}_{3}&=\Pr(J_{0}=1)\pi_{2}+\Pr(J_{0}=2)\pi_{1}&=1/2+1/6=2/3,\\
&\vdots&
\end{eqnarray*}
\end{exam}

\subsection{Joint inclusion probabilities}

Joint inclusion probabilities of a renewal chain sampling design can be derived from the Probability Mass Function (PMF) $f(.)$ of $J_{1}$. Indeed, the selection of unit $\ell$ given that unit $k$, $0<k<\ell$, is selected can be decomposed according to the number of selected units between $k$ and $\ell$, and this number does not depend on $J_{0}$. We can write that:
$$
\Pr(\ell\in S| k\in S ) =  \sum_{j=1}^{\ell-k} f^{j*}(\ell-k).
$$
The joint inclusion probability of units $k\neq \ell$ is thus given by~(\ref{pikl}).
\begin{equation}\label{pikl}
\pi_{k\ell}= \pi_{k} \sum_{j=1}^{\ell-k} f^{j*}(\ell-k), \ k<\ell .
\end{equation}

\subsection{Bernoulli sampling}\label{secbern}

The Bernoulli sampling design with inclusion probabilities $\pi$ is obtained by selecting or not units into the sample through independent Bernoulli trials with parameter $\pi$ (see for example \citet[][p.43]{til:06}). Its probability distribution is given by
$$
P(s) = \pi^{n}(1-\pi)^{N-n},\ s \subset U,
$$
where $n=\# s$ is the size of sample $s$. The joint inclusion probabilities are equal to $\pi_{k\ell}=\pi^{2}$,  $k\neq \ell$. The usual algorithm used to select a sample according to the Bernoulli sampling design simply consists of generating $N$ independent Bernoulli variables and selecting units according to the observed values.

Bernoulli sampling can also be implemented using Definition~\ref{def1}. Indeed, it is clear that spacings of a Bernoulli sampling design are i.i.d. distributed variables with shifted geometric distributions $J_{i}=1+X_{i}$ where
$
\Pr(X_{i}=k)=(1-\pi)^{k}\pi,\ k\geq 0.
$
Bernoulli sampling is thus a simple renewal chain sampling design satisfying Definition~\ref{def1}. On the other hand it is easy to prove that, if $X_{i}$ follows a geometric distribution, then $X_{i}$ has the same distribution as its forward transform (it is the only distributions on $\N$ that enjoy this property). The random variable $J_{0}$ of Definition~\ref{nj0} has the same distribution as $J_{1}$ in this particular case. Consequently, Bernoulli sampling is also an equilibrium renewal chain sampling design according to Definition~\ref{def3}.

Using~(\ref{pikl}), we find the second order inclusion probabilities $\pi_{k\ell}=\pi^2,$ $k\neq \ell$. Indeed, the sum of $j$ i.i.d. geometric random variables with parameter $\pi$ follows a negative binomial distribution with parameters $j$ and $\pi$. The negative binomial distribution with parameters $j \geq 1$ and $\pi$ in $(0,1)$ is defined by its PMF:
\begin{equation}\label{defnb}
f_{\mathpzc{NB}}(x)=\binom{j+x-1}{x} (1-\pi)^{x}\pi^{j} ,\ x\in \N,
\end{equation}
where $\binom{a}{b}=a!/[b!(a-b)!]$ if $b\le a$ are non-negative integers, and $\binom{a}{b}=0$ if $a$, $b$ or $a-b$ is negative. Considering that
$$
\sum_{i=1}^{j}J_{i}=j+\sum_{i=1}^{j}X_{i},\ j\geq 1,
$$
we have that
$$
f^{j*}(x)=\binom{x-1}{x-j}(1-\pi)^{x-j}\pi^{j}, \ x\geq j.
$$
From~(\ref{pikl}), we get that the joint inclusion probabilities are equal to:
\begin{eqnarray*}
\pi_{k\ell}
&=& \pi \sum_{j=1}^{\ell-k} f^{j*}(\ell-k), \ k<\ell, \\
&=& \pi \sum_{j=1}^{\ell-k} \binom{\ell-k-1}{\ell-k-j}(1-\pi)^{\ell-k-j}\pi^{j},\\
&=& \pi^{2}(\pi+1-\pi) ^{\ell-k-1}=\pi^{2}.
\end{eqnarray*}

\subsection{Systematic sampling}\label{secsyst}

Systematic sampling with rate $1/r$, $r\in \N^{*}$, from a population $U=\{1,\dots,N\}$ is obtained by generating a random start $u$ with a uniform discrete distribution between $1$ and $r$, and selecting units $k$ of $U$ such that $k \equiv u \pmod{r}$ into the sample (see \citet[][]{mad:mad:44}). The first-order inclusion probabilities of this sampling design are given by $\pi_{k}=1/r,$ $k\in U,$ and its joint inclusion probabilities by
$$
\pi_{k\ell}=1/r \mbox{ if } k \equiv \ell \pmod{r} \mbox{ and } 0 \mbox{ otherwise}.
$$
If $N=m r$, with $m,r\in \N^{*}$, the sample size is deterministic and equal to $m$.

Systematic sampling is an equilibrium renewal chain sampling design, agreeing with Definition~\ref{def3} where the $J_{i}$'s, $i\geq 1$ are deterministic and equal to $r$. Indeed, the forward transform $X_{F}$ of $X=r-1$, $r\in \N^{*}$ is such that:
$$
\Pr(X_{F}=k) = \frac{\Pr(X\ge k)}{\E(X+1)}= \frac{\1_{\{r-1 \geq k \}}}{r},\ k\in \N,
$$
and $X_{F}$ follows a uniform distribution on $\{0,\dots,r-1\}$. Hence $J_{0}$ follows a uniform distribution on $\{1,\dots,r\}$. 

The joint inclusion probabilities are obtained from~(\ref{pikl}). Indeed, the sum of $j$ spacings $J_{i}$, $i\geq 1$  is deterministic, equal to $j r$, and
$$
f^{j*}(x)=\1_{\{j r = x\}},\ x\geq 1.
$$
We then have that, for $k<\ell$,
\begin{eqnarray*}
\pi_{k\ell}
&=& \pi \sum_{j=1}^{\ell-k} f^{j*}(\ell-k)=\frac{1}{r} \sum_{j=1}^{\ell-k} \1_{\{j r =\ell-k\}}
= \frac{1}{r} \1_{\{ k \equiv \ell  \pmod{r} \}}.
\end{eqnarray*}
We confirm with this expression that most of the joint inclusion probabilities are null, making it impossible to estimate the variance of Horvitz-Thompson estimators without bias.

\section{Spreading renewal chain sampling designs}\label{spreading}

We have seen in Sections~\ref{secbern} and~\ref{secsyst} two examples of renewal chain sampling designs with very different spreading properties. In Bernoulli sampling, the selection of units are independent, even if they are adjacent in the population list. In systematic sampling, the selection of adjacent units is impossible, provided that the sampling rate is smaller than 1. This translates to the variance of the spacings distribution: it is null for systematic sampling, that has perfect spreading properties, and it is quite large, equal to $(1-\pi)/\pi^{2}$, for Bernoulli sampling.

Using Definition~\ref{def3}, we can build sampling designs with any given spacing distribution on $\N^{*}$. The expectation of this distribution is forced by the sampling rate, which is usually itself decided as a function of cost or precision constraints.  In Section~\ref{qsbern}, we give an application with shifted negative binomial spacings, allowing for a limited control on the variance and spreading properties of the design. As a limiting case, we find the shifted Poisson spacings of Section~\ref{limbern}.
To have a variance that is arbitrarily small, in Section~\ref{qsbin} we use  shifted binomial distributions that have a variance always smaller than their expectation.

These are only examples, and any distribution or family of distributions on $\N^{*}$ that offers sufficient control on its shape can be used. Table~\ref{discrete} in Appendix contains a list of useful discrete probability distributions with their probability mass functions, their supports, means and variances.

\subsection{Negative binomial spacings}\label{qsbern}

The definition of the negative binomial distribution in~(\ref{defnb}) can be extended to parameters $r>0$ and $p$ in $(0,1)$ by considering the PMF:
$$
f_{\mathpzc{NB}(r,p)}(x)=\frac{\Gamma(r+x)}{x!\Gamma(r)} p^{r} (1-p)^{x},\ x\in \N,
$$
where $\Gamma(r) = \int_{0}^{+\infty} t^{r-1}e^{-t}\ dt$, $r>0$ and $\Gamma(k)=(k-1)!$, $k\in \N^{*}$. The expectation of this distribution is $r(1-p)/p$ and its variance is $r(1-p)/p^{2}$.

We consider equilibrium renewal sampling designs with positive spacings $J_{i}$, $i\geq 1$ such that $J_{i}-1$ follows a negative binomial distribution with parameters $r$ and $p$, denoted $\mathpzc{NB}(r,p)$. For a given sampling rate $\pi\in (0,1)$, we find that $\E(J_i)=1/\pi$ implies that
$$
p=\frac{r\pi}{r\pi+ 1-\pi }.
$$
It follows that
\begin{equation}\label{varnb}
\Var(J_i) =\frac{1-\pi}{\pi}+\frac{1}{r}\left(\frac{1-\pi}{\pi}\right)^{2}.
\end{equation}
When $r=1$, we find, as a special case, the Bernoulli sampling design. From~(\ref{varnb}), we deduce that the variance of spacings is smaller than that of Bernoulli sampling when $r>1$ and in that case there is a repulsion between selected units: the sample is spread more evenly on the population than if drawings were independent. On the contrary, if $r<1$, there is an attraction between units and selecting neighbor units together is more likely.

The sum of $j$ independent random variables with negative binomial distribution and parameters $r$, $p$, has a negative binomial distribution with parameters $j r$ and $p$. 
\begin{proposition}\label{pik2nb}
The second order inclusion probabilities of an equilibrium renewal chain sampling design with shifted negative binomial spacings, sampling rate $\pi$ and first parameter $r$ are
$$
\pi_{k\ell}= \pi \sum_{j=1}^{\ell-k} \frac{\Gamma(j r +\ell-k-j)}{(\ell-k-j)! \Gamma(j r )} p^{j r} (1-p)^{\ell-k-j}, \ k< \ell,
$$
where $p=r\pi/(r\pi+ 1-\pi )$.
\end{proposition}
\begin{proof}
Since $J_{i}-1$ has a $\mathpzc{NB}(r,p)$ distribution, $\left(\sum_{i=1}^{j} J_{i}\right)-j$ has a  $\mathpzc{NB}(jr,p)$ distribution.
Using~(\ref{pikl}) one gets the result.
\end{proof}
These joint inclusion probabilities remain positive for any value of $r$. They are plotted in Figure~\ref{G3} for $\pi=1/30$ and different values of $r$.
\begin{figure}[htb!]
\begin{center}
\includegraphics[scale=0.5]{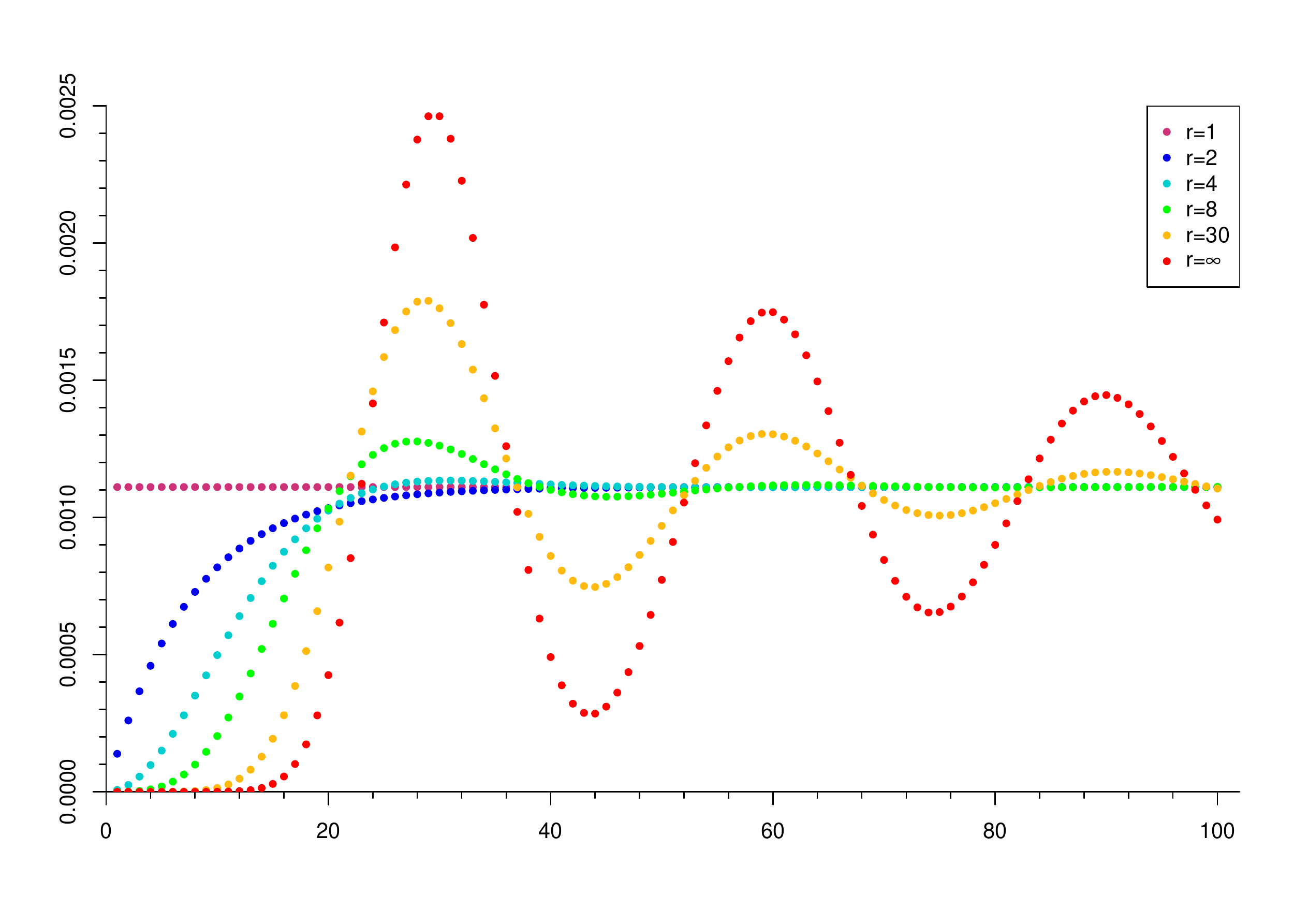}
\caption{\label{G3}$\pi_{k,k+i}$ in function of $i$ for negative binomial spacings with $\pi=1/30$, $r=1,2,4,8,30$ and $r=+\infty$ (Poisson spacings). When $r=1$, we obtain the Bernoulli sampling design and a flat line on the plot. Oscillations are stronger for larger values of $r$.}
\end{center}
\end{figure}

In order to have an equilibrium renewal chain sampling design and equal first order inclusion probabilities, the first sample unit index has to be generated from a shifted forward negative binomial. We get that $J_{0}-1 \sim \mathpzc{ForNB}(r,p)$, where the definition of $\mathpzc{ForNB}(r,p)$ can be found in Table~\ref{discrete}.

\subsection{Poisson spacings}\label{limbern}

The limit of negative binomial distributions when $r$ tends to infinity and $p$ tends to 0 while keeping a constant expectation $\lambda=r(1-p)/p$, is a Poisson distribution $\mathpzc{P}(\lambda)$ with parameter $\lambda$ which is also its expectation and its variance.

We consider the equilibrium renewal chain sampling design with shifted Poisson spacings: $J_{i}-1\sim \mathpzc{P}(\lambda)$, where $\lambda=(1-\pi)/\pi$, $i\in \N^{*}$. The first inter arrival $J_{0}$ is selected using a shifted forward Poisson distribution: $J_{0}-1 \sim \mathpzc{ForP}(\lambda)$, where the definition of distribution $\mathpzc{ForP}(\lambda)$ can be found in Table~\ref{discrete}.

\begin{proposition}\label{pik2poiss}
The second order inclusion probabilities of an equilibrium renewal chain sampling design with shifted Poisson spacings and sampling rate $\pi$ are, with $\lambda=(1-\pi)/\pi$
$$
\pi_{k\ell} =  \pi \sum_{j=1}^{\ell -k}  \frac{e^{-j \lambda} \left(j \lambda\right)^{\ell -k-j}}{(\ell -k-j)!}, \ k< \ell.
$$
\end{proposition}

\subsection{Binomial spacings}\label{qsbin}

The variance of spacings in Sections~\ref{qsbern} and~\ref{limbern} are bounded from below, by $(1-\pi)/\pi$. However, to get a sample spread close to that of systematic sampling, we need to be able to have a variance that is arbitrarily close to 0. For this, we consider the equilibrium renewal chain sampling design that is obtained with shifted binomial spacings: $J_{i}-1 \sim \mathpzc{Bin}(r,p)$, $i\in \N^{*}$, $r \in \N$, $p\in [0,1]$. The first spacing $J_{0}$ is selected using a shifted forward binomial distribution: $J_{0}-1 \sim \mathpzc{ForBin}(r,p)$, where the definition of distribution $\mathpzc{ForBin}(r,p)$ can be found in Table~\ref{discrete}.

We find that with a sampling rate equal to $\pi$, $r$ must  necessarily be greater or equal to $(1-\pi)/\pi$, and that
$$
p=\frac{1}{r}\left(\frac{1-\pi}{\pi}\right).
$$
The variance of spacings is then given by~(\ref{varbin}),
\begin{equation}\label{varbin}
\Var(J_{i}) =\frac{1-\pi}{\pi}-\frac{1}{r}\left(\frac{1-\pi}{\pi}\right)^{2}, \ i\in \N^{*}.
\end{equation}

Considering the constraints on $r$ and $p$, this variance is minimal when $r$ is the smallest integer that is greater or equal to $(1-\pi)/\pi$. With this $r$, the variance of spacings is always smaller than $1$, which is really small for an integer valued random variable with a usually very large expectation $1/\pi$. When $1/\pi$ is an integer, the variance of spacings is null when $r=(1-\pi)/\pi$ and $p=1$. The sampling design obtained then is just the systematic sampling design.

If $r$ tends to infinity and $p=(1-\pi)/r\pi$, the binomial distribution with parameters $r$ and $p$ converges in distribution toward the Poisson distribution with parameter $(1-\pi)/\pi$. Hence, the sampling design of Section~\ref{limbern} is also the limiting case of Binomial spacings renewal chain sampling design when $r$ tends to infinity.

\begin{proposition}\label{pik2bin}
The second order inclusion probabilities of an equilibrium renewal chain sampling design with shifted binomial spacings, sampling rate $\pi$ and first parameter $r$ are equal to
$$
\pi_{k\ell}= \pi \sum_{j=1}^{\ell-k} \binom{jr}{\ell-k-j} p^{\ell-k-j} \left(1-p\right)^{j (r+1)-\ell-k}, \ k < \ell,
$$
where $p=(1-\pi)/r\pi$.
\end{proposition}

\subsection{Summary}

The different renewal chain sampling designs we considered are listed in Table~\ref{varchain} with the variance of their spacings.
\begin{table}[htb!]
\begin{center}
\caption{Renewal chain sampling designs and variance of their spacings.}\label{varchain}
\begin{tabular}{ll}
\hline
Distribution of $J_{1}-1$                     & $\Var(J_{1})$\\
\hline
 & \\
Negative binomial, $r>0$           & $\frac{1-\pi}{\pi}+\frac{1}{r}\left(\frac{1-\pi}{\pi}\right)^{2} $ \\[3mm]
Poisson                                       & $\frac{1-\pi}{\pi}                                               $ \\[3mm]
Binomial,  $r\geq (1-\pi)/\pi$                                    & $\frac{1-\pi}{\pi}-\frac{1}{r}\left(\frac{1-\pi}{\pi}\right)^{2} $ \\[3mm]
Systematic or binomial with $r=(1-\pi)/\pi$ & $0$                                                                \\
\hline
\end{tabular}
\end{center}
\end{table}
If $1/\pi$ is not an integer, the variance of spacings cannot be null and is at least $(\left\lceil 1/\pi\right\rceil-1/\pi)(1/\pi-\left\lfloor 1/\pi \right\rfloor)$. This lower bound is not reached with shifted binomial spacings but the binomial renewal chain sampling design enjoys the desirable property of having positive joint inclusion probabilities. Other spacing distributions can be used but we retained only common families of distribution that have useful properties such as stability under convolution.

\section{Fixed size sampling designs with exchangeable circular spacings}\label{circular}

Except in very special situations, renewal chain sampling designs do not have fixed sample size. This is due to the independence of spacings. However, in many applications fixed size is required. In this section, we propose to define sampling designs using exchangeable instead of independent spacings. We obtain fixed size designs with equal inclusion probabilities, and we are able to control the sample spread by the choice of the random spacings distribution.

\subsection{Circular spacings}
A sampling design of fixed size $n$ in a population $U=\{1,\dots,N\}$ is entirely specified by the joint distribution of one of the unit indexes, e.g. $X_{1}$, and the ``circular'' spacings $J_{i}=X_{i+1}-X_{i}$, $i=1,\dots,n-1$, and $J_{n}=N+X_{1}-X_{n}$, where $X_{1}$ is the smallest sample unit index and $X_{n}$ the largest. If we represent the population $U$ around a table, as in Figure~\ref{fig:circle}, the $J_i$'s are the difference of units position. Note that considering the population as circular is not new in survey sampling and goes back at least to \citet{ful:70}.

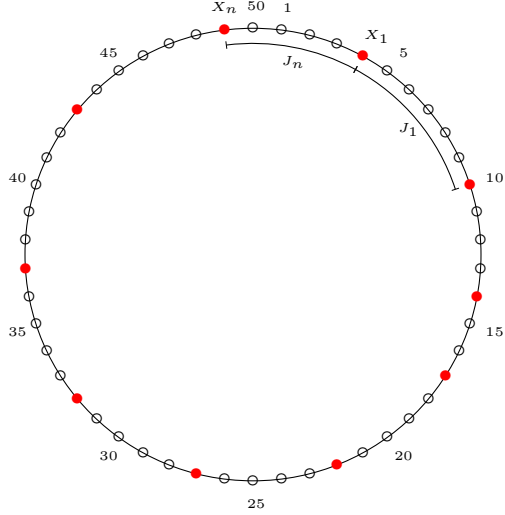
\begin{figure}[htb!]
\begin{center}
\begin{tikzpicture}
%


	\def\N{50}
 	\draw (0:0) circle [radius = 3];
	\foreach \x in {1,2,3,5,6,7,8,9,11,12,13,15,16,18,19,20,21,23,24,25,26,28,29,30,31,33,34,35,36,38,39,40,41,42,44,45,46,47,48,50} {
         \draw ({90-360/ \N *\x}  :3) node {$\circ$} ;
  };
  \foreach \x in {1,5,10,...,\N} {
       \draw ({90-360/ \N *\x} : 3.3) node{\begin{tiny} \x\end{tiny}} ;
  };
  
	\foreach \x in {4,10,14,17,22,27,32,37,43,49} {
       \draw[color = red] ({90-360/ \N *\x}  :3) node {$\bullet$} ;
  };
	\draw ({90-360/ \N *4}    :3.3) node{\begin{tiny} $X_{1}$\end{tiny}};
	\draw ({90-360/ \N *49}    :3.3) node{\begin{tiny} $X_{n}$\end{tiny}};
	
	\draw ({90-360/ \N *49} :2.8) arc ({90-360/ \N *-1}:{90-360/ \N *4} :2.8) ;
	\draw [line width = 0.1, color = black] ({90-360/ \N *49} : 2.75) -- ({90-360/ \N *49} : 2.85);
	\draw [line width = 0.1, color = black] ({90-360/ \N *4} : 2.75) -- ({90-360/ \N *4} : 2.85);
 	\draw ({90-360/ \N *1.5}  :2.6) node{\begin{tiny} $J_{n}$\end{tiny}};

 	\draw ({90-360/ \N *4}  :2.8) arc ({90-360/ \N *4} :{90-360/ \N *10} :2.8) ;
	\draw [line width = 0.1, color = black] ({90-360/ \N *4} : 2.75) -- ({90-360/ \N *4} : 2.85);
	\draw [line width = 0.1, color = black] ({90-360/ \N *10} : 2.75) -- ({90-360/ \N *10} : 2.85);
 	\draw ({90-360/ \N *7}  : 2.6) node{\begin{tiny} $J_{1}$\end{tiny}};

\end{tikzpicture}
\end{center}
\caption{\label{fig:circle} Illustration of a sample of 10 units, in a population of $50$ units. The selected units are in red.}
\end{figure}

We intend to work with Equation~(\ref{samplecirc}) that defines without loss of generality the random sample $S$ of a fixed size sampling design,
\begin{equation}\label{samplecirc}
S=\left\{S_{j} \pmod{N},\ j=1,\dots,n \right\},
\end{equation}
where $J_{1},\dots,J_{n}$ are positive integer random variables that sum to $N$, $J_{0}$ is a random variable in $U$, and
$$
S_{j}=\sum_{i=0}^{j}J_{i}.
$$

\subsection{First order inclusion probabilities}

The first order inclusion probabilities of a sampling design that results from~(\ref{samplecirc}) depend on the joint distribution of the $J_{i}$'s, $i=0,\dots,n$. Intuitively, one sees that, for a given joint distribution of $J_{1},\dots,J_{n}$ with $S_{n}=N$, choosing $J_{0}$ to be independent of the other $J_{i}$'s and uniform on $\{1,\dots,N\}$ allows to obtain equal first order inclusion probabilities. To prove this assertion, one can compute in the general case the inclusion probability of a unit $k$. Consider an independent $J_{0}$, and let 
$$
f_{0}(t)=\Pr(J_{0}=t),\  f_{j}(k)=\Pr(S_{j} \pmod{N}=k),\ t,\ k \in U,\ 1\leq j \leq n.
$$ 
By conditioning on the event $\{J_0=t\}$ and using the law of total probability on the disjoint events $\{k>t\}$, $\{k=t\}$ and $\{k<t\}$, we get that:
$$
\pi_{k}=\sum_{t=1}^{N}f_{0}(t)\left[\1_{t=k}+\1_{t<k}\sum_{j=1}^{k-t}f_{j}(k-t)+\1_{t>k}\sum_{j=1}^{N+k-t}f_{j}(N+k-t)\right],
$$
with the convention that $f_{j}(k)=0$ if $j>n$. Hence we can write that vectors $\pig=(\pi_{1},\dots,\pi_{N})$ and $\fb_{0}=(f_{0}(1),\dots,f_{0}(N))$ are solutions of the linear equation
\begin{equation}\label{lineq}
\pig=\Ab \fb_{0},
\end{equation}
where $\Ab$ is the square matrix of size $N$ with general term
\begin{equation}\label{eqpikcirc}
a_{kt}=\1_{t=k}+\1_{t<k}\sum_{j=1}^{k-t}f_{j}(k-t)+\1_{t>k}\sum_{j=1}^{N+k-t}f_{j}(N+k-t),\ 1\leq k,t \leq N.
\end{equation}
It is not our purpose to solve the system and find designs with any given inclusion probabilities, especially since solutions depend on the $f_{j}(k)$'s, but one arrives rapidly to the result that all the lines of $\Ab$ sum to $n$ (see Proposition~\ref{appcircpik} in Appendix). Hence, if $f_{0}(t)=1/N$ for all $t$, then the inclusion probabilities are all equal to $n/N$.

\subsection{Joint inclusion probabilities}


Let $\ell > k$ in $\{1,\dots,N\}$, and consider $\Pr(\ell\in S| k\in S)$ so that, by definition, $\pi_{k\ell}=\pi_{k}\Pr(\ell\in S| k\in S)$. Knowing that $k\in S$, the event $\ell \in S$ can be decomposed according to the number of units selected between $k$ and $\ell$:
\begin{equation}\label{eqdiese}
\pi_{k\ell}=\pi_{k}\sum_{j=1}^{\ell-k}\Pr\left(\ell \in S \mbox{ and } \# S \cap \{k+1,\dots,\ell\} =j | k\in S \right).
\end{equation}
The term $\Pr\left(\ell \in S \mbox{ and } \# S \cap \{k+1,\dots,\ell\} =j  | k\in S \right)$ is usually difficult to compute: one must decompose according to which $S_{i}$ is equal to $k$. However, if the joint distribution of $J_{1},\dots,J_{n}$ has some additional properties, as in Proposition~\ref{proppiklcirc}, one can obtain a simple expression. 

\begin{proposition}\label{proppiklcirc}
Consider a positive integer random vector $(J_{1},\dots,J_{n})$ that sums to $N$ and such that the distributions of any sum of $k$ successive $J_{i}$'s are equal, this condition also holding for the ``circular'' sums of $J_{n-i}$ up to $J_{n}$ and $J_{1}$ up to $J_{k-i}$ if $i<k$. Then, the second order inclusion probabilities are given by
\begin{equation}\label{piklfs}
\pi_{k\ell}=\pi_{k}\sum_{j=1}^{\ell-k}f_{j}(\ell-k),\ k<\ell.
\end{equation}
\end{proposition}
\begin{proof}
Indeed, we then have that:
$$
\Pr\left(\ell \in S \mbox{ and } \# S \cap \{k+1,\dots,\ell\} =j | k\in S\right)=f_{j}(\ell-k),\ j=1,\dots,\ell-k,
$$
and the result follows immediately. 
\end{proof}

In the situation of Proposition~\ref{proppiklcirc}, we also get that the conditional inclusion probability $\Pr(\ell\in S| k\in S)$ is a function of $\ell-k \pmod{N}$:
$$
\Pr(\ell\in S| k\in S)=\sum_{j=1}^{\ell-k}f_{j}(\ell-k),\mbox{ if } k<\ell, 
$$
$$
\Pr(\ell\in S| k\in S)=\sum_{j=1}^{N+\ell-k}f_{j}(N+\ell-k),\mbox{ if }  \ell<k. 
$$
It also follows in that case that if the first order inclusion probabilities are all equal, for example when $J_{0}$ has a uniform distribution, then the joint inclusion probabilities $\pi_{k\ell}$ depend only on $\ell - k \pmod{N}$.


Actually, all distributions considered for spacings $J_{1},\dots,J_{n}$ in this paper enjoy a stronger property, they are exchangeable distributions \citep{ald:85,kal:05}.
\begin{definition}\label{exchangeable}
A family $J_{1},\dots,J_{n}$ of random variables is said to be exchangeable if, for all $1\leq k \leq n$ and permutation $\sigma$ of $\{1,\dots,n\}$, the joint distribution of $(J_{\sigma(1)},\dots J_{\sigma(k)})$ is equal to the joint distribution of $(J_{1},\dots,J_{k})$. If the $J_{i}$'s are discrete distributions, this is equivalent to say that $\Pr(J_{1}=a_{1},\dots,J_{n}=a_{n})$ is a symmetric function of $(a_{1},\dots,a_{n})$.
\end{definition}
Exchangeable integer distributions $J_{1},\dots,J_{n}$ that sum to $N$ clearly satisfy the conditions of Proposition~\ref{proppiklcirc}. They are the natural equivalents to the i.i.d. spacing distributions used in Section~\ref{renewal} when the spacings are constrained, by the fixed sample size, to sum to $N$.
\begin{definition}\label{defcircsamp}
Fixed size sampling designs with exchangeable circular spacings and uniform inclusion probabilities are the sampling designs with random samples $S=\left\{S_{j} \pmod{N},\ j=1,\dots,n \right\}$ where $S_{j}=\sum_{i=0}^{j}J_{i}$, $J_{0}$ is a uniform random distribution on $\{1,\dots,N\}$ independent from $\Jb=(J_{1},\dots,J_{n})$, and the $J_{i}$'s, $i=1,\dots,n$, are exchangeable positive integer distributions that sum to $N$.
\end{definition}
The PMF of a fixed size sampling design with exchangeable circular spacings and uniform inclusion probabilities $J_{1},\dots,J_{n}$ is given simply by
\begin{equation}\label{circsamppmf}
P(s)=\frac{n}{N}\Pr(J_{1}=x_{2}-x_{1},\dots,J_{n-1}=x_{n}-x_{n-1},J_{n}=N+x_{1}-x_{n}),
\end{equation}
where $x_{1},\dots,x_{n}$ are the ordered indexes of units sampled in $s$. Indeed, $P(s)$ can be decomposed according to the value of $J_{0}$ into:
\begin{eqnarray*}
\Pr(S=\{x_{1},\dots,x_{n}\})
&=&\Pr(J_{0}=x_{1})\Pr(J_{1}= x_{2}-x_{1},\dots,,J_{n-1}= x_{n}-x_{n-1},J_{n}=N+x_{1}-x_{n})\\
&+& \Pr(J_{0}=x_{2})\Pr(J_{1}= x_{3}-x_{2},\dots,J_{n-1}= N+x_{1}-x_{n},J_{n}=x_{2}-x_{1})\\
&\vdots&\\
&+& \Pr(J_{0}=x_{n})\Pr(J_{1}= N+x_{1}-x_{n},\dots,J_{n}=x_{n}-x_{n-1}),
\end{eqnarray*}
and the $\Pr(J_{0}=x_{i})$ are all equal to $1/N$ while the probabilities involving $J_{1},\dots,J_{n}$ are all equal due to the exchangeability of the circular spacings.
\subsection{Simple Random Sampling}\label{SRSsec}

The Simple Random Sampling (SRS) without replacement design of fixed sample size $n$ is defined by:
$$
P(s) = \binom{N}{n}^{-1} \mbox{ if }\# s =n \mbox{, and } P(s) = 0 \mbox{ otherwise}.
$$
A SRS sample can be selected using the following algorithm (\citet{fan:mul:rez:62}, see also \citet[][p. 46]{til:06}): define a counter $j=0$, then, for $k=1$ to $N$, select unit $k$ with probability $(N-j)/(N-k-1)$ and update $j=j+1$ if $k$ is selected. It is also possible to obtain this design by generating successive jumps according to negative hypergeometric distributions with parameters that depend on the previously selected units (see \citet{vit:84,vit:85,vit:87}).

Proposition~\ref{propsrs} asserts that SRS is a sampling design with exchangeable circular spacings, where the spacings follow a shifted multivariate negative hypergeometric distribution. The (singular) multivariate negative hypergeometric distribution (see for example \citet[][pp. 171-199]{joh:kot:bal:97}) of size $n\geq 1$, with parameters $m\in \N$ and $\rb=(r_1,\dots,r_n)$, $r_{i}>0$, $i=1,\dots,n$ is a probability distribution on integer vectors $(x_{1},\dots,x_{n})$ that sum to $m$. It is denoted here by $\mathpzc{MNH}(m,\rb)$, and has a PMF given by:
\begin{equation}\label{multneghyp}
f_{\mathpzc{MNH}(m,\rb)}(x_{1},\dots,x_{n}) = \frac{m!\Gamma{(R)}}{\Gamma{(m+R)}} \prod_{i=1}^{n}   \frac{\Gamma{(r_{i}+x_{i})}}{\Gamma{(r_{i})}x_{i}!},
\end{equation}
where $R=\sum_{i=1}^{n} r_{i}$.
\begin{proposition}\label{propsrs}
SRS is the sampling design of~(\ref{samplecirc}), where $J_{0}$ has a uniform distribution on $U$, is independent of the $J_{i}$'s, $i\geq 1$, and the integer random vector $\Jb=(J_{1},\dots,J_{n})$ follows a shifted multivariate negative hypergeometric distribution: $\Jb-{\bf 1}_{n} \sim \mathpzc{MNH}(N-n,{\bf 1}_{n})$, where ${\bf 1}_{n}$ is the $n-$vector of ones.
\end{proposition}
\begin{proof}
With parameter $N-n$ and ${\bf 1}_{n}$, the PMF given in~(\ref{multneghyp}) reduces to
$$
f_{\mathpzc{MNH}(N-n,{\bf 1}_{n})}(x_{1},\dots,x_{n})=\binom{N-1}{n-1}^{-1},
$$
where $x_{1},\dots,x_{n}$ are non-negative integers that sum to $N-n$. Hence, $\Jb$ has a uniform distribution on the vectors of positive integer numbers that sum to $N$,
$$
\Pr\left[\Jb=(j_{1},\dots,j_{n})\right]=\binom{N-1}{n-1}^{-1},
$$
for all positive integers $(j_{1},\dots,j_{n})$ that sum to $N$. Moreover, this PMF is symmetric in its arguments and the $J_{i}$'s are exchangeable. Applying~(\ref{circsamppmf}), we get that
$$
\Pr(S=\{x_{1},\dots,x_{n}\})=\frac{n}{N} \binom{N-1}{n-1}^{-1}=\binom{N}{n}^{-1},
$$
for all $x_{1}<\dots<x_{n}$.
\end{proof}

The marginal distributions of the circular spacings are shifted negative hypergeometric distributions. Indeed, the marginal distributions of a $\mathpzc{MNH}(m,\rb)$-distributed vector are negative hypergeometric distributions (see for example \citet[][]{jan:pat:72}) with respective parameters $m$, $r_{i}$, $R=\sum_{j}r_{j}$, and PMF
$$
f_{\mathpzc{NH}(m,r_{i},R)}(x)=\frac{m!\Gamma{(R)}}{\Gamma{(m+R)}} \frac{\Gamma{(r_{i}+x)}}{\Gamma{(r_{i})}x!} \frac{\Gamma{(R-r_{i}+m-x)}}{\Gamma{(R-r_{i})}(m-x)!} , \ x \leq m, \ x\in \N.
$$
Their expectation and variance are, respectively,  $mr_{i}/R$ and $m(r_{i}/R)(1-r_{i}/R)(R+m)/(R+1)$. It follows that $J_{k}-1$ has a negative hypergeometric distribution with parameters $N-n$, $1$, and $n$,  $k=1,\dots,n$. In particular we have that
$$
\E(J_{k})=\frac{N-n}{n}+1=\frac{N}{n},
$$
$$
\Var(J_{k})=\frac{N-n}{n}\left(1-\frac{1}{n}\right)\frac{N}{n+1}.
$$

The second order inclusion probabilities can be derived from~(\ref{piklfs}). Indeed, the sum of components of a multivariate negative hypergeometric distribution follows a negative hypergeometric distribution (see \citet[][]{jan:pat:72}). Its parameters are derived from the parameters $m$ and $\rb$ by summing the $r_{i}$'s that correspond to the components that are in the sum. Hence we have that
$$
\sum_{i=1}^{j}J_{i}=j+K_{j},
$$
where $K_{j}$ follows a negative hypergeometric distribution with parameters $N-n$, $j$ and $n$. We can deduce that
\begin{eqnarray*}
f_{j}(\ell-k)&=&\frac{(N-n)!\Gamma{(n)}}{\Gamma{(N)}} \frac{\Gamma{(j+\ell-k-j)}}{\Gamma{(j)}(\ell-k-j)!} \frac{\Gamma{(n-j+N-n-\ell+j+k)}}{\Gamma{(n-j)}(N-n-\ell+j+k)!},\\
&=&\frac{(N-n)!(n-1)!(\ell-k-1)!(N-\ell+k-1)!}{(N-1)!(j-1)!(\ell-k-j)!(n-j-1)!(N-n-\ell+k+j)!},\\
&=& \frac{n-1}{N-1}\binom{N-2}{\ell-k-1}^{-1}\binom{n-2}{j-1}\binom{N-n}{\ell-k-j},
\end{eqnarray*}
and that
$$
\pi_{k\ell}=\frac{n(n-1)}{N(N-1)}\sum_{j=1}^{\ell-k}\binom{N-2}{\ell-k-1}^{-1}\binom{n-2}{j-1}\binom{N-n}{\ell-k-j},\ k<\ell,
$$
via~(\ref{piklfs}).
However, if we rename $u=j-1$, $v=\ell-k-1$, $t=n-2$ and $s=N-2$, this last sum is
$$
\binom{s}{v}^{-1}\sum_{u=0}^{v}\binom{t}{u}\binom{s-t}{v-u},
$$
and Vandermonde's identity ensures that it is equal to 1. Hence we find the well known result:
$$
\pi_{k\ell}=\frac{n(n-1)}{N(N-1)}, \ k\neq \ell.
$$

\subsection{Systematic sampling}

If $N=rn$ with $r\in \N$, the systematic sampling design presented in Section~\ref{secsyst} is a fixed size sampling design with exchangeable circular spacings. It is trivially obtained by taking $J_{0}$ uniform on $U$ and $J_{i}=r$, $i=1,\dots,n$. The joint inclusion probabilities can also easily be derived from~(\ref{piklfs}) using that $f_{j}(\ell-k)=\1_{\{\ell=k+jr\}}$.

\section{Spreading fixed size sampling designs with exchangeable circular spacings}\label{spreading2}

Similar to what we did in Section~\ref{spreading}, we introduce in Sections~\ref{neghyp},~\ref{multin} and~\ref{hyp} new sampling designs with spreading properties by choosing different circular spacings distributions.

Following the structure of Section~\ref{spreading}, we work, in Section~\ref{neghyp}, on sampling designs with multivariate negative hypergeometric spacings and a spreading control parameter $r>0$. When $0<r<1$, there is an attraction between the selected units: if a unit is selected, then its neighbors are more likely to be selected. If $r=1$, the design is SRS and if $r>1$, the sampling is better spread than SRS. As a limit case when $r$ is large, we obtain the multinomial circular spacings design of Section~\ref{multin}. The spacings variance of these sampling designs is bounded from below. Smaller variances and better spreading properties are obtained with multivariate hypergeometric circular spacings in Section~\ref{hyp}, furthering the parallel with binomial spacings of Section~\ref{qsbin}.

\subsection{Multivariate negative hypergeometric circular spacings}\label{neghyp}

The multivariate negative hypergeometric distribution $\mathpzc{MNH}(m,\rb)$ has exchangeable marginals exactly when $r_{1}=\dots=r_{n}$, $\rb=r{\bf 1}_{n}$ for some positive real number $r$. If $\Jb-{\bf 1}_{n} \sim \mathpzc{MNH}(N-n,r{\bf 1}_{n})$, the sampling design of Definition~\ref{defcircsamp} has circular spacings with a variance given by:
$$
\Var(J_{k})=\frac{N-n}{n}\left(1-\frac{1}{n}\right)\frac{rn+N-n}{rn+1}, \ k=1,\dots,n.
$$
These variances are decreasing functions of $r$, with SRS corresponding to $r=1$.

According to~(\ref{circsamppmf}), the sampling design PMF is given by:
$$
P(s)=\frac{n}{N}\frac{\Gamma(nr)}{\left[\Gamma(r)\right]^{n}}\frac{(N-n)!}{ \Gamma[N+n(r-1)]}\frac{\Gamma(r+N+x_{1}-x_{n}-1)}{(N+x_{1}-x_{n}-1)!}\prod_{i=1}^{n-1}\frac{\Gamma(r+x_{i+1}-x_{i}-1)}{(x_{i+1}-x_{i}-1)!},
$$
where $x_{1},\dots,x_{n}$ are the ordered indexes of units sampled in $s$. The second order inclusion probabilities are obtained from~(\ref{piklfs}):
$$
\pi_{k\ell}=\frac{n}{N}\sum_{j=1}^{\ell-k}\binom{N-n}{\ell-k-j}\frac{\B[\ell-k+j(r-1),N+n(r-1)-\ell+k-j(r-1)]}{\B(jr,nr-jr)},\ k<\ell,
$$
where $\B(\cdot,\cdot)$ denotes the beta function, defined by:
$$
\B(a,b)=\frac{\Gamma(a)\Gamma(b)}{\Gamma(a +b)}=\int_{0}^{1} t^{a-1} (1-t)^{b-1}dt,
$$
if $a$ and $b$ are positive real numbers.


These joint inclusion probabilities are plotted in Figure~\ref{G4} for different values of $r$, including their limit when $r\rightarrow\infty$. On this plot, we see strong oscillations of the joint inclusion probabilities when $r$ is large.
\begin{figure}[htb!]
\begin{center}
\includegraphics[scale=0.5]{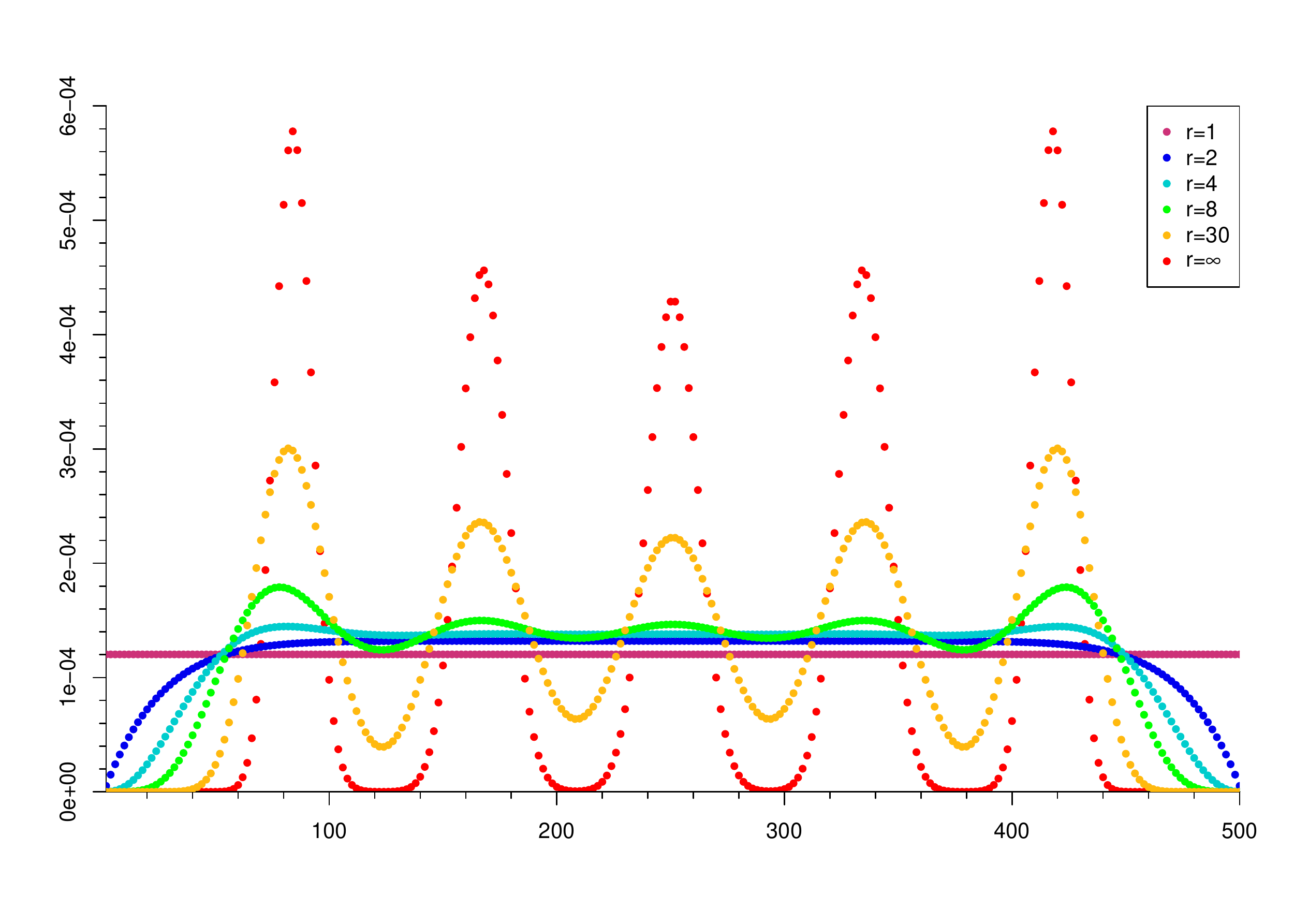}
\caption{\label{G4}$\pi_{1,1+i}$ in function of $i$ for a fixed size sampling design with shifted multivariate negative hypergeometric circular spacings, $n=6$, $N=500$ and $r=1,\ 2,\ 4,\ 8$, $30$ and $r=+\infty$ (multinomial circular spacings).
When $r=1$, we obtain the SRS design and constant joint inclusion probabilities. The larger $r$ is, the more contrasted are the joint inclusion probabilities.}
\end{center}
\end{figure}

\subsection{Multinomial circular spacings}\label{multin}

Let $\mathpzc{MNom}(m,\pb)$ be the multinomial distribution with parameters $m\in \N$ and $\pb=(p_{1},\dots,p_{n})\in [0,1]^{n}$, $\sum_{i}p_{i}=1$. It is the probability distribution on integer vectors $(x_{1},\dots, x_{n})$ such that $\sum_{i} x_{i}=m$ with PMF:
\begin{equation}\label{multom}
f_{\mathpzc{MNom}(m,\pb)}(x_{1},\dots,x_{n}) = \binom{m}{x_{1} \cdots x_{n}} \prod_{i=1}^{n} p_{i}^{x_{i}},
\end{equation}
where
$$
\binom{m}{x_{1} \cdots x_{n}}=\frac{m!}{x_{1}! \cdots x_{n}!}.
$$
Its marginal distributions are binomial with respective parameters $m$ and $p_{i}$. They are exchangeable exactly when $\pb=n^{-1}{\bf 1}_{n}$

When $r$ tends to infinity, the multivariate negative hypergeometric distribution with parameters $m$ and $r{\bf 1}_{n}$ tends to a multinomial distribution with parameters $m$ and ${\bf 1}/n$ (see, for instance, \citet[][p. 182]{terrell1999mathematical}). Thus the fixed size sampling design with shifted multinomial exchangeable circular spacings is the limit case of multivariate negative hypergeometric spacings of Section~\ref{neghyp} when $r$ tends to infinity.

Spacings $J_{k}$ follow a shifted binomial distribution with parameters $N-n$ and $1/n$.
The corresponding sampling design PMF is given by
$$
P(s) = \frac{n}{N}\frac{1}{n^{N-n}}\frac{(N-n)!}{(N+x_{1}-x_{n}-1)! \prod_{j=2}^{n} (x_{j}-x_{j-1}-1)! },
$$
where $x_{1},\dots,x_{n}$ are the ordered indexes of units sampled in $s$. The sum of any $j$ components of a $\mathpzc{MNom}(m,\pb)$ multinomial vector follows a binomial distribution with parameters $m$ and $p$ where $p$ is the sum of corresponding $p_{i}$'s. Hence we have here that
$$
f_{i}(\ell-k)=\binom{N-n}{\ell-k-j}\left(\frac{j}{n}\right)^{\ell-k-j}\left(1-\frac{j}{n}\right)^{N-n-\ell+k+j},
$$
$$
\pi_{k\ell}
=
\frac{n}{N} \sum_{j=1}^{\ell-k}
\binom{N-n}{\ell-k-j}\left(\frac{j}{n}\right)^{\ell-k-j}\left(1-\frac{j}{n}\right)^{N-n-\ell+k+j},\ k<\ell.
$$

\subsection{Multivariate hypergeometric circular spacings}\label{hyp}

Variances of circular spacings in Sections~\ref{neghyp} and~\ref{multin} are bounded from below. In order to have smaller variances, one can use shifted multivariate hypergeometric spacings.

Let $\mathpzc{MH}(m,\rb)$ be the (singular) multivariate hypergeometric distribution with parameters $m \in \N$ and $\rb=(r_{1},\dots,r_{n})$, $\rb$ is an integer vector that sums to $R$ and $m\leq R$. It is a probability distribution on integer vectors $(x_{1},\dots,x_{n})$ that sum to $m$, $x_{i}\leq r_{i}$ and has a  PMF  given by:
\begin{equation}\label{multghyp}
f_{\mathpzc{MH}(m,\rb)}(x_{1},\dots,x_{n}) = \binom{R}{m}^{-1}\prod_{i=1}^{n} \binom{r_{i}}{x_{i}}.
\end{equation}
The marginal distributions of $\mathpzc{MH}(m,\rb)$ are hypergeometric variables with respective parameters $m$, $r_{i}$ and $R$, and their variance is $m(r_{i}/R)(1-r_{i}/R)(R-m)/(R-1)$.

The multivariate hypergeometric distribution has exchangeable marginals exactly when $\rb=r{\bf 1}_{n}$ for some integer $r$ larger than $m/n$. Here, we consider the fixed size sampling design with circular spacings $\Jb$ such that $\Jb-{\bf 1}_{n}\sim \mathpzc{MH}(N-n,r{\bf 1}_{n})$ with $r\geq N/n-1$. We get that the spacing variances are given by
$$
\Var(J_{k})= \frac{N-n}{n}\left(1-\frac{1}{n}\right)\frac{rn-N+n}{rn-1}.
$$
The parameter $r$ can be used to tune the variance. If $r= (N-n)/n$ is integer, we have $\Var(J_{k})=0$ and we obtain the systematic sampling design.

The sampling design PMF is obtained via~(\ref{circsamppmf}):
$$
P(s) = \frac{n}{N}\binom{rn}{N-n}^{-1}\binom{r}{x_{i+1}-x_{i}-1}\prod_{i=1}^{n-1} \binom{r}{N+x_{1}-x_{n}-1},
$$
where $x_{1},\dots,x_{n}$ are the ordered indexes of units sampled in $s$. The sum of components of a $\mathpzc{MH}(m,\rb)$ vector follows a hypergeometric distribution. In the present case we find that
$$
f_{j}(x)=\binom{rn}{N-n}^{-1}\binom{jr}{x-j}\binom{rn-jr}{N-n-x+j},\ j\leq x\leq N-n,
$$
and that the joint inclusion probabilities are
$$
\pi_{k\ell}
=
\frac{n}{N}\binom{rn}{N-n}^{-1} \sum_{j=1}^{\ell-k}
\binom{jr}{\ell-k-j}\binom{rn-jr}{N-n-\ell+k+j}, \ k <\ell.
$$
Note that some joint inclusion probabilities may be null, even when $r>(N-n)/n$ and the design is not the systematic sampling design. For example, if $N=10$, $n=2$, $r=5$ and $\ell=k+1$, then
$$
\pi_{k\ell}=\frac{1}{5}\binom{10}{8}^{-1}
\binom{5}{0}\binom{5}{8}=0.
$$

\subsection{Summary}

The different fixed size sampling designs with exchangeable circular spacings that we considered are listed in Table~\ref{varfixed} with the variance of their spacings.
\begin{table}[htb!]
\begin{center}
\caption{Fixed size sampling designs and variance of their spacings.}\label{varfixed}
\begin{tabular}{ll}
\hline
Distribution of   $\Jb -{\bf 1}_{n}$ &  $\Var(J_{i}),\ i=1,\dots,n$\\
\hline
 & \\
Multivariate negative hypergeometric, $r>0$ & $\frac{N-n}{n}\left(1-\frac{1}{n}\right)\frac{rn+N-n}{rn+1}$ \\[3mm]
Multinomial                                & $\frac{N-n}{n}\left(1-\frac{1}{n}\right)$                         \\[3mm]
Multivariate hypergeometric, $r\geq N/n-1$  & $\frac{N-n}{n}\left(1-\frac{1}{n}\right)\frac{rn-N+n}{rn-1}$ \\[3mm]
Systematic or hypergeometric with $r=N/n-1$ & $0$ \\
\hline
\end{tabular}
\end{center}
\end{table}
Other exchangeable circular spacings may be used, and are easily obtained as distributions of vectors of i.i.d. random variables conditioned on the sum of the vectors components. The families of distribution of Section~\ref{spreading2} encompass the SRS and fixed-size systematic designs. They allow one to use designs with low spacings variance, but it is not always possible to avoid having null joint inclusion probabilities with shifted multivariate hypergeometric spacings. Finally, the designs are not strictly sequential as the population list may need to be run over twice in order to finish selecting a sample.

\section{Simulations}\label{simulations}
A single artificial population of $N=200$ units was generated with an interest variable $Y$ that had a trend and are autocorrelated:
$
y_{k} = k + z_{k},
$
where $z_{k}=0.6 z_{k-1}+\varepsilon_{k}$
and $\epsilon_{k} \sim N(0,\sigma_{\varepsilon} = 0.3)$.
With this kind of autocorrelation, having well spread samples ought to be an efficient strategy. The ``spacing'' $N+x_{1}-x_{n}$ between the last sampled unit and the first one is treated like any other spacing, so that ideally one would also want to have some similarity between units at the beginning of the population list and those at the end. This feature can easily be obtained in a setting of continuous population sampling (see \citet[][]{wilh:qual:till:2017}), but is not common in finite population applications.

For each situation, a set of 100,000 samples was generated. All samples were of fixed size $n=50$ and were selected using the following sampling designs:
\begin{itemize}
\item  Multivariate negative hypergeometric (MNH) with $r=0.5$, $r=1$ (SRS), $r=5$, $r=10$, $r=50$,
\item  Multinomial  (MULT),
\item  Multivariate hypergeometric (MH) with $r=50$, $r=10$, $r=6$, $r=4$.
\end{itemize}
We used different values for the tunings parameter $r$ in all kinds of sampling design in order to show the effect of this tuning parameter.  

For each sample an estimate $\widehat{\ybar}$ of the mean and of the variance $\widehat{\Var}_{SYG}\left(\widehat{\ybar}\right)$ (with the Sen-Yates-Grundy formula) were produced. Compiling our simulation results, we computed the following values, presented in Table~\ref{tabsimulations}:
\begin{itemize}
\item the Bias Ratio
$$
{\rm BR} = 100\frac{\E_{sim}\left(\widehat{\ybar} - \ybar\right)}{\left[\Var_{sim}\left(\widehat{\ybar}\right)\right]^{\frac{1}{2}}},
$$
where $\E_{sim}(\cdot)$ and $\Var_{sim}(\cdot)$ denote the empirical means and variances of the simulation results.
\item the standard error
$$
{\rm SE} = \left[\Var_{sim}\left(\widehat{\ybar}\right)\right]^{\frac{1}{2}};
$$
\item
the square root of the variance estimator average
$$
{\rm REVAR} = \left\{\E_{sim}\left[ \widehat{\Var}_{SYG}\left(\widehat{\ybar}\right)\right]\right\}^{\frac{1}{2}};
$$
\item the coefficient of variation of the variance estimator
$$
{\rm CV}
= \frac{\left\{\Var_{sim}\left[ \widehat{\Var}_{SYG}\left(\widehat{\ybar}\right)\right]\right\}^{\frac{1}{2}} }{ {\Var}_{sim}\left(\widehat{\ybar}\right)};
$$
\item and the coverage rate of the $95\%$ confidence interval.
\end{itemize}

The simulation results in Table~\ref{tabsimulations} confirm, with column BR, that the estimator of the mean is unbiased. The accuracy of the mean estimator improves as the circular spacings variance decreases, from Design {\rm MNH} $r=0.5$ to Design {\rm NH} $r=4$.

The conclusions are different for the variance estimator. For all situations in our simulations, the joint inclusion probabilities are positive. The variance estimator is unbiased, and this is confirmed by the fact that columns SE and REVAR are mostly equal. However, when the variance of the spacings are close to 0, the variance estimator is unstable. Indeed, with these parameters, some joint inclusion probabilities are very small (less than $1/1000$) compared to others (on average $0.0625$). In column CV that the accuracy of the variance estimator improves at first as the circular spacings variance decreases, from Design {\rm MNH} $r=0.5$ to Design {\rm MNH} $r=5$, and then the coefficient of variation goes up again from Design {\rm MNH} $r=0.5$ to Design {\rm MH} $r=4$. The coverage rate deviates strongly from its nominal value of $95\%$ in the last couple of designs.
\begin{table}[htb!]
\centering
\caption{\label{tabsimulations} Results of the 100,000 simulations. The designs are ordered in decreasing order of the variance of the spacings.}
\vspace*{0.2cm}
\begin{tabular}{lrrrrrr}
  \hline
 & BR & SE & REVAR & CV & coverage \\
  \hline
  {\rm MNH} $r=0.5$ & -0.25 & 0.46 & 0.45 & 0.48 & 93.97 \\ 
  {\rm SRS} & -0.12 & 0.35 & 0.35 & 0.23 & 94.52 \\ 
  {\rm MNH} $r=5$ & 0.08 & 0.23 & 0.23 & 0.21 & 94.39 \\ 
  {\rm MNH} $r=10$ & -0.22 & 0.21 & 0.21 & 0.26 & 94.08 \\ 
  {\rm MNH} $r=50$ & 0.13 & 0.19 & 0.19 & 0.33 & 93.90 \\ 
  {\rm MULT} & 0.36 & 0.19 & 0.19 & 0.35 & 93.64 \\ 
  {\rm MH} $r=50$ & -0.17 & 0.18 & 0.18 & 0.37 & 93.58 \\ 
  {\rm MH} $r=10$ & -0.35 & 0.16 & 0.16 & 0.52 & 92.05 \\ 
  {\rm MH} $r=6$  & -0.74 & 0.14 & 0.14 & 0.72 & 83.97 \\ 
  {\rm MH} $r=4$ & -0.52 & 0.11 & 0.15 & 1.60 & 40.55 \\ 
   \hline
\end{tabular}
\end{table}
Thus the design that performs best for the point estimation of the mean does not allow to properly estimate the precision, and even gives seriously misleading confidence intervals. The same kind of problem arises when a stratified sampling design is used with too many strata.

An arbitration needs to be made between the accuracy of the point estimator and that of its variance estimator. In our simulations, a reasonable solution consists in choosing the sampling design with shifted multinomial distribution (MULT). This method is simple to implement, more so than the MNH or MH. It allows for accurate point estimation while presenting a correct coverage rate of its confidence intervals.

\section{Conclusions}\label{conclusion}

In Sections~\ref{renewal} and~\ref{circular}, we proposed general methods to generate uniform inclusion probabilities sampling designs with i.i.d. or exchangeable spacings. We used them in Sections~\ref{spreading} and~\ref{spreading2} to obtain sample selection methods with controlled spreading properties and gave, in Section~\ref{simulations}, an example where such methods are useful. If the response variable is similar among units that are close in the population list, the choice of the spreading parameter allows one to make a trade-off between precision of the point estimator and precision of variance estimator. 

Some of the designs that we consider  have concentrated spacings, but, unlike systematic sampling, they retain positive joint inclusion probabilities and thus allow for an unbiased estimation of variance. These joint inclusion probabilities have computable closed-form expressions and depend only on the ``distance'' between units in the population list, thus at most $N-1$ joint inclusion probabilities need to be computed. However, the ranks of sampled units in the population must be known in order to compute a variance estimator.

We do not have a clear solution to extending these results in all generality to unequal first order inclusion probabilities sampling designs. One partial solution is to work on the distribution of $J_{0}$. The choice of a different distribution for $J_{0}$ allows one to have a limited control on the inclusion probabilities via Equations~\ref{pirenew} and~\ref{lineq}, while leaving the spacings untouched. Another solution is the thinning approach. It consists in selecting a large enough first phase sample with a spreading design and uniform inclusion probabilities and selecting a second phase sub-sample with appropriate inclusion probabilities. However, this does not preserve the spreading properties.

\section*{Acknowledgments}
This work was supported in part by the Swiss Federal Statistical Office. The views expressed in this paper are solely those of the authors.~M. W. was partially supported by a Doc.Mobility fellowship of the Swiss National Science Foundation (grant no. P1NEP2\_162031). The authors would like to thank three referees and an associate editor for their constructive comments that helped us improve this paper, and Prof. Lennart Bondesson who has kindly sent us copies of his work on a similar topic.


\section*{Appendix A: Proof of Proposition~\ref{fondamental_disc} and Remark~\ref{remarque}}

\begin{lemma}\label{lemme1}
If $f(\cdot)$ is a probability distribution on $\{1,2,\dots,\}$ with cumulative distribution function $F(\cdot)$, and $k,j\geq 1$, then
$$
\sum_{t=1}^k f^{(j+1)*}(t)=\sum_{t=1}^k  f^{j*}(t) F(k-t).
$$
\end{lemma}
\begin{proof}
Indeed, if $\1_{A}$ is the indicator function of set $A$,
\begin{eqnarray*}
\sum_{t=1}^k f^{(j+1)*}(t) &=& \sum_{t=1}^k \sum_{u=1}^t f^{j*}(u) f(t-u),\\
                           &=& \sum_{t}\sum_{u} f^{j*}(u) f(t-u) \1_{\{1\leq u\leq t\}}\  \1_{\{1\leq t\leq k\}},\\
                           &=& \sum_{u}\sum_{t}  f^{j*}(u) f(t-u) \1_{\{1\leq u\leq k\}}\  \1_{\{1\leq u\leq t\}}\  \1_{\{1\leq t\leq k\}},\\
                           &=& \sum_{u}  f^{j*}(u) \1_{\{1\leq u\leq k\}} \left[ \sum_{t} f(t-u) \ \1_{\{1\leq u\leq t\}} \1_{\{1\leq t\leq k\}}\right], \\
                           &=& \sum_{u=1}^{k}  f^{j*}(u)\left[F(k-u) - \underbrace{F(0)}_{=0}\right],\\
													 &=& \sum_{t=1}^k  f^{j*}(t) F(k-t).
\end{eqnarray*}
\end{proof}

\begin{proof}[Proof of Proposition~\ref{fondamental_disc}]
$f_{0}(\cdot)$ is a well-defined non-negative function on $\mathbb{N}$. It is sufficient to prove that $\sum_{k\geq 0} f(\{k+1,\dots\}) = \mu$, but
\begin{eqnarray*}
\sum_{k\geq 0} f(\{k+1,\dots\}) &=& \sum_{k\geq 0} \sum_{j\geq k+1} f(j),\\
                                &=& \sum_{j\geq 0} \sum_{k\geq 0} f(j) \1_{k+1 \leq j},\\
																&=& \sum_{j\geq 0}  j \cdot f(j),\\
																&=& \mu.
\end{eqnarray*}
As $f_{0}(k-t)=\left[1-F(k-t)\right]/\mu$, to prove~(\ref{flat}), it is sufficient to note that
\begin{eqnarray*}
\sum_{t=1}^{k} \left[1-F(k-t)\right] \sum_{j=1}^{t} f^{j*}(t) &=& \sum_{t}\sum_{j}\left[1-F(k-t)\right]f^{j*}(t)\1_{1\leq t \leq k}\1_{1\leq j \leq t}, \\
                                                   &=& \sum_{j}\sum_{t}\left[1-F(k-t)\right]f^{j*}(t)\1_{1\leq t \leq k}\1_{1\leq j \leq t}, \\
                                                   &=& \sum_{j} \left[ \sum_{t}f^{j*}(t)\1_{1\leq t \leq k}\1_{1\leq j \leq t} - \sum_{t}F(k-t)f^{j*}(t)\1_{1\leq t \leq k}\1_{1\leq j \leq t}\right],\\
\end{eqnarray*}
\begin{eqnarray*}																\phantom{\sum_{t=1}^{k} \left[1-F(k-t)\right] \sum_{j=1}^{t} f^{j*}(t)} 				 &=& \sum_{j=1}^{k} \left[  \sum_{t=j}^{k}f^{j*}(t) -  \sum_{t=j}^{k} F(k-t)f^{j*}(t)\right],\\
																									 &=& \sum_{j=1}^{k} \left[ \sum_{t=1}^{k}f^{j*}(t) -  \sum_{t=1}^{k} F(k-t)f^{j*}(t)\right]\mbox{ (indeed, $f^{j*}(t)=0$ if $t<j$)},\\
																									 &=& \sum_{t=1}^{k}f^{1*}(t) - \sum_{t=1}^{k} F(k-t)f^{k*}(t) \mbox{ via lemma~\ref{lemme1}},\\
																									 &=& F(k) - \sum_{t=1}^{k} f^{(k+1)*}(t) = F(k),
\end{eqnarray*}
since $f^{(k+1)*}(t)=0$ if $t\leq k$, and the result follows immediately.
\end{proof}

\begin{proof}[Proof of Remark~\ref{remarque}]
Consider $X$ a random variable on $\N$ with finite moment of order $m+1$, $\E(X^{m+1})$, $m\geq 0$, and its forward transform $X_{F}$ according to Definition~\ref{forward}. Then we can write: 
\begin{eqnarray*}
\sum_{k\geq 0}k^{m}\Pr(X_{F}=k) & = & \sum_{k\geq 0}k^{m}\frac{\Pr(X \geq k)}{\E(X+1)} = \sum_{k\geq 0}\sum_{i\geq k}\frac{k^{m} \Pr(X=i)}{\E(X+1)},\\
   &=& \frac{1}{\E(X+1)}\sum_{i\geq 0}\sum_{k \geq 0}\1_{k\leq i} k^{m}\Pr(X=i) = \frac{1}{\E(X+1)}\sum_{i\geq 0} \left(\sum_{k=0}^{i}k^{m}\right)\Pr(X=i),\\
	 &=& \frac{\E[F_{m}(X)]}{\E(X+1)},
\end{eqnarray*} 
where  $F_{m}(x)=\sum_{k=0}^{x}k^{m}$.
\end{proof}

\section*{Appendix B}

\begin{proposition}\label{appcircpik}
The lines of matrix $\Ab$ with general term $a_{kt}$ given at~(\ref{eqpikcirc}) all sum to $n$.
\end{proposition}
\begin{proof}
We have
$$
a_{kt}=\1_{t=k}+\1_{t<k}\sum_{j=1}^{k-t}f_{j}(k-t)+\1_{t>k}\sum_{j=1}^{N+k-t}f_{j}(N+k-t),
$$
with $f_{j}(t)=0$ if $j<t$, $t\leq 1$, $t>N$ or $j>n$. We also have that $f_{n}(N)=1$ and $f_{j}(N)=0$ if $j<n$. The conclusion follows from
\begin{eqnarray*}
\sum_{t=1}^{N}\1_{t<k}\sum_{j=1}^{k-t}f_{j}(k-t)&=&\sum_{t=1}^{N}\sum_{j=1}^{N}f_{j}(k-t)\1_{j\leq k-t}\1_{j \leq n},\\
  &=&\sum_{j=1}^{n}\sum_{t=1}^{N}f_{j}(k-t)=\sum_{j=1}^{n}\Pr(S_{j}\leq k-1),\mbox{ and}\\
\sum_{t=1}^{N}\1_{t>k}\sum_{j=1}^{N+k-t}f_{j}(N+k-t)&=& \sum_{t=1}^{N}\sum_{j=1}^{N}f_{j}(N+k-t)\1_{j\leq N+k-t}\1_{j\leq n}\1_{t>k},\\
  &=&\sum_{j=1}^{n}\sum_{t=1}^{N}f_{j}(N+k-t)\1_{t>k}=\sum_{j=1}^{n}\left[\Pr(S_{j}\geq k)-f_{j}(N)\right].
\end{eqnarray*}
\end{proof}

\section*{Appendix C: discrete probability distributions}
Let $\R_{+}$ denote the set of positive real numbers,
   $$\Gamma(r,x) = \int_{x}^{+\infty} t^{r-1}e^{-t}\ dt, \quad \gamma(r,x) = \int_0^{x} t^{r-1}e^{-t}\ dt,$$
where $r>0,x>0$  and
$$
\B_x(a,b)=\int_0^x t^{a-1} (1-t)^{b-1}dt ,
\hspace{1cm}
\I_x(a,b)=\frac{\B_x(a,b)}{\B(a,b)},
$$
with $a>0,b>0,0<x<1.$

\begin{table}[htb!]
\begin{center}
\caption{Discrete distributions of probability\label{discrete}}
\footnotesize
\begin{tabular}{p{1.5cm}lccccc}
\hline
 Name
        &  Notation
        &  PMF
        &  Support
        &  Parameters
        &  Mean
        &  Variance \\[3.1mm]
\hline
 Bernoulli
        & $ \mathpzc{Bern}(p) $
        & $p^{x} (1-p)^{1-x} $
        & $ \{0,1\}$
        & $ p\in [0,1] $
        & $ p $
        & $ p(1-p) $ \\[3.1mm]
Forward Bernoulli
        & $ \mathpzc{ForBern}(p) $
        & $  \frac{p^{x}}{p+1} $
        & $ \{0,1\} $
        & $ p\in [0,1] , n\in \N $
        & \multicolumn{2}{c}{(see below the table)}\\[3.1mm]
Binomial
        & $ \mathpzc{Bin}(n,p) $
        & $ \binom{n}{x} p^{x} (1-p)^{n-x} $
        & $ \{0,\dots,n\} $
        & $ p\in [0,1] , n\in \N $
        & $ np $
        & $ np(1-p) $ \\[3.1mm]
Forward Binomial
        & $ \mathpzc{ForBin}(n,p) $
        & $ \frac{\I_{p}(x,n-x+1)}{np+1}$
        & $ \{0,\dots,n\}$
        & $ p\in [0,1] , n\in \N $
        & \multicolumn{2}{c}{(see below the table)}\\[3.1mm]
 Geometric
        & $ \mathpzc{G}(1-p) $
        & $ p(1-p)^{x} $
        & $ \N $
        & $ p\in [0,1] $
        & $ \frac{1-p}{p} $
        & $ \frac{1-p}{p^{2}} $ \\[3.1mm]
Negative Binomial
        & $ \mathpzc{NB}(r,p) $
        & $ \frac{\Gamma(r+x)}{x! \Gamma(r)} p^{r} (1-p)^{x} $
        & $ \N $
        & $ p\in [0,1] , r\in \N^{*} $
        & $ \frac{r(1-p)}{p} $
        & $ \frac{r(1-p)}{p^{2}} $ \\[3.1mm]
 Forward  Negative Binomial
        & $ \mathpzc{ForNB}(r,p) $
        & $ \frac{p\I_{(1-p)}(x,r)}{r(1-p)+p}$
        & $ \N $
        & $ p\in [0,1] , r\in \N^{*} $
        &\multicolumn{2}{c}{(see below the table)}\\[3.1mm]
 Poisson
        & $ \mathpzc{P}(\lambda) $
        & $ \frac{e^{-\lambda} \lambda^x}{x!} $
        & $ \N $
        & $ \lambda\in \R_{+} $
        & $ \lambda $
        & $ \lambda $ \\[3.1mm]
 Forward Poisson
        & $ \mathpzc{ForP}(\lambda) $
        & $ \frac{1}{\lambda+1}\left[\1_{x=0}+\frac{\gamma(x,\lambda)}{(x-1)!}\1_{x\geq 1}\right]$
        & $ \N$
        & $ \lambda\in \R_{+} $
        & \multicolumn{2}{c}{(see below the table)}\\[5.1mm]
 Hypergeo\-metric
        & $ \mathpzc{H}(m,r,R) $
        & $ \frac{\binom{r}{x}\binom{R-r}{m-x}}{\binom{R}{m}} $
        & $ \begin{array}{c}
				\{0,\dots,m\}\bigcap
				\\ \{r+m-R,\dots,r\}
	          \end{array}$
        & $ \begin{array}{c} m,r,R \in \N^{*},\\ m,r\le R \end{array} $
        & $ \frac{mr}{R} $
        & $ \frac{mr(R-r)}{R^{2}}\frac{R-m}{R-1} $ \\[5.1mm]
Negative  Hypergeometric
        & $ \mathpzc{NH}(m,r,R) $
        & $ \frac{\frac{\Gamma{(r+x)}}{\Gamma{(r)}x!}\frac{\Gamma{(R-r+m-x)}}{\Gamma{(R-r)}(m-x)!}}{\frac{\Gamma{(m+R)}}{\Gamma{(R)}m!}}$
        & $ \{0,\dots, m \} $
        & $ \begin{array}{c} m,r,R\in\N^{*}\\ 1\le R-r\end{array} $
        & $ \frac{mr}{R} $
        & $ \frac{mr(R-r)}{R^{2}}\frac{R+m}{R+1} $ \\[3.1mm]
Uniform
        & $ \mathpzc{U}(0,a) $
        & $ \frac{1}{a+1}$
        & $ \{0,\dots, a \} $
        & $ a \in\N$
        & $ \frac{a}{2}$
        & $ \frac{ (a+1)^2-1}{12} $ \\[3.1mm]
\hline
\end{tabular}\\
Expectations and variances of forward distributions are easily computed in function of the first three moments of the original distribution (see Remark~\ref{remarque}).
\end{center}
\end{table}

\newpage



\clearpage
\begin{thebibliography}{}

\bibitem[Aldous, 1985]{ald:85}
Aldous, D. (1985).
\newblock Exchangeability and related topics.
\newblock In Hennequin, P., editor, {\em École d'Été de Probabilités de
  Saint-Flour XIII — 1983}, volume 1117 of {\em Lecture Notes in
  Mathematics}, pages 1--198. Springer Berlin Heidelberg.

\bibitem[Barbu and Limnios, 2008]{bar:lim:08}
Barbu, V. and Limnios, N. (2008).
\newblock {\em Semi-Markov Chains and Hidden Semi-Markov Models Toward
  Applications: Their Use in Reliability and DNA Analysis}.
\newblock Springer, New York.

\bibitem[Bellhouse, 1988]{bel:88b}
Bellhouse, D.~R. (1988).
\newblock Systematic sampling.
\newblock In {Krishnaiah}, P.~R. and Rao, C.~R., editors, {\em Handbook of
  Statistics Volume 6: Sampling}, pages 125--145, Amsterdam.
  Elsevier/North-Holland.

\bibitem[Bellhouse and Rao, 1975]{bel:rao:75}
Bellhouse, D.~R. and Rao, J. N.~K. (1975).
\newblock Systematic sampling in the presence of a trend.
\newblock {\em Biometrika}, 62(3):694--697.

\bibitem[Bellhouse and Sutradhar, 1988]{Bell:Sutr:vari:1988}
Bellhouse, D.~R. and Sutradhar, B.~C. (1988).
\newblock Variance estimation for systematic sampling when autocorrelation is
  present.
\newblock {\em The Statistician}, 37(3):327--332.

\bibitem[Bondesson, 1986]{Bondesson86}
Bondesson, L. (1986).
\newblock Sampling of linearly ordered population by selection of units at
  successice random distances.
\newblock Technical Report~25, Swedish University of agricultural sciences,
  Section of forest biometry, Ume\aa.

\bibitem[Bondesson and Thorburn, 2008]{bon:tho:08}
Bondesson, L. and Thorburn, D. (2008).
\newblock A list sequential sampling method suitable for real-time sampling.
\newblock {\em Scandinavian Journal of Statistics}, 35(3):466--483.

\bibitem[Chauvet, 2012]{chauvet2012characterization}
Chauvet, G. (2012).
\newblock On a characterization of ordered pivotal sampling.
\newblock {\em Bernoulli}, 18(4):1320--1340.

\bibitem[Cochran, 1946]{coc:46}
Cochran, W.~G. (1946).
\newblock Relative accuracy of systematic and stratified random samples for a
  certain class of population.
\newblock {\em Annals of Mathematical Statistics}, 17(2):164--177.

\bibitem[Cox, 1962]{cox1962renewal}
Cox, D.~R. (1962).
\newblock {\em Renewal Theory}.
\newblock Methuen, London.

\bibitem[Daley and Vere-Jones, 2002]{Daley02}
Daley, D. and Vere-Jones, D. (2002).
\newblock {\em An Introduction to the Theory of Point Processes: Volume I:
  Elementary Theory and Methods}.
\newblock Springer, New York.

\bibitem[Deville, 1998]{dev:98a}
Deville, J.-C. (1998).
\newblock Une nouvelle (encore une!) m\'ethode de tirage \`a probabilit\'es
  in\'egales.
\newblock Technical Report 9804, M\'ethodologie Statistique, INSEE, Paris.

\bibitem[Deville and Till\'e, 1998]{dev:til:98}
Deville, J.-C. and Till\'e, Y. (1998).
\newblock Unequal probability sampling without replacement through a splitting
  method.
\newblock {\em Biometrika}, 85(1):89--101.

\bibitem[Fan et~al., 1962]{fan:mul:rez:62}
Fan, C.~T., Muller, M.~E., and Rezucha, I. (1962).
\newblock Development of sampling plans by using sequential (item by item)
  selection techniques and digital computer.
\newblock {\em Journal of the American Statistical Association},
  57(298):387--402.

\bibitem[Feller, 1971]{fel:71}
Feller, W. (1971).
\newblock {\em An introduction to Probability Theory and its applications}.
\newblock Wiley, New-York.

\bibitem[Fuller, 1970]{ful:70}
Fuller, W.~A. (1970).
\newblock Sampling with random stratum boundaries.
\newblock {\em Journal of the Royal Statistical Society: Series B (Statistical
  Methodology)}, 32(2):209--226.

\bibitem[Grafstr\"om, 2010]{Grafstrom2010982}
Grafstr\"om, A. (2010).
\newblock On a generalization of poisson sampling.
\newblock {\em Journal of Statistical Planning and Inference}, 140(4):982 --
  991.

\bibitem[Horvitz and Thompson, 1952]{hor:tho:52}
Horvitz, D.~G. and Thompson, D.~J. (1952).
\newblock A generalization of sampling without replacement from a finite
  universe.
\newblock {\em Journal of the American Statistical Association},
  47(260):663--685.

\bibitem[Iachan, 1982]{iac:82}
Iachan, R. (1982).
\newblock Systematic sampling: A critical review.
\newblock {\em International Statistical Review}, 50(3):293--303.

\bibitem[Iachan, 1983]{iac:83}
Iachan, R. (1983).
\newblock Asymptotic theory of systematic sampling.
\newblock {\em Annals of Statistics}, 11(3):959--969.

\bibitem[Janardan and Patil, 1972]{jan:pat:72}
Janardan, K.~G. and Patil, G.~P. (1972).
\newblock A unified approach for a class of multivariate hypergeometric models.
\newblock {\em Sankhy$\bar{a}$: The Indian Journal of Statistics, Series A
  (1961-2002)}, 34(4):363--376.

\bibitem[Johnson et~al., 2005]{joh:kem:kot:05}
Johnson, N.~L., Kemp, A.~W., and Kotz, S. (2005).
\newblock {\em Univariate Discrete Distributions}.
\newblock Wiley, New York.

\bibitem[Johnson et~al., 1997]{joh:kot:bal:97}
Johnson, N.~L., Kotz, S., and Balakrishnan, N. (1997).
\newblock {\em Discrete Multivariate Distributions}.
\newblock Wiley, New York.

\bibitem[Kallenberg, 2005]{kal:05}
Kallenberg, O. (2005).
\newblock {\em Probabilistic Symmetries and Invariance Principles}.
\newblock Springer, New York.

\bibitem[{Loonis} and {Mary}, 2015]{2015arXiv151006618L}
{Loonis}, V. and {Mary}, X. (2015).
\newblock {Determinantal Sampling Designs}.
\newblock {\em ArXiv e-prints}.

\bibitem[Madow and Madow, 1944]{mad:mad:44}
Madow, L.~H. and Madow, W.~G. (1944).
\newblock On the theory of systematic sampling.
\newblock {\em Annals of Mathematical Statistics}, 15(1):1--24.

\bibitem[Madow, 1949]{mad:49}
Madow, W.~G. (1949).
\newblock On the theory of systematic sampling, {II}.
\newblock {\em Annals of Mathematical Statistics}, 20(3):333--354.

\bibitem[Meister, 2004]{Meister04}
Meister, K. (2004).
\newblock {\em On methods for real time sampling and distributions in
  sampling}.
\newblock PhD thesis, Department of Mathematical Statistics, Ume\aa \
  University.

\bibitem[Mitov and Omey, 2014]{Mitov14}
Mitov, K.~V. and Omey, E. (2014).
\newblock {\em Renewal Processes}.
\newblock Springer, New York.

\bibitem[Murthy and Rao, 1988]{mur:rao:88}
Murthy, M.~N. and Rao, T.~J. (1988).
\newblock Systematic sampling with illustrative examples.
\newblock In Krishnaiah, P.~R. and Rao, C.~R., editors, {\em Handbook of
  Statistics Volume 6: Sampling}, pages 147--185, Amsterdam.
  Elsevier/North-Holland.

\bibitem[Pea et~al., 2007]{pea:qua:til:07}
Pea, J., Qualit\'e, L., and Till\'e, Y. (2007).
\newblock Systematic sampling is a minimal support design.
\newblock {\em Computational Statistics \& Data Analysis}, 51(12):5591--5602.

\bibitem[Sen, 1953]{sen:53}
Sen, A.~R. (1953).
\newblock On the estimate of the variance in sampling with varying
  probabilities.
\newblock {\em Journal of the Indian Society of Agricultural Statistics},
  5:119--127.

\bibitem[Terrell, 1999]{terrell1999mathematical}
Terrell, G.~R. (1999).
\newblock {\em Mathematical Statistics: A Unified Introduction}.
\newblock Springer, New York.

\bibitem[Till\'e, 1996]{til:96c}
Till\'e, Y. (1996).
\newblock A moving stratification algorithm.
\newblock {\em Survey Methodology}, 22:85--94.

\bibitem[Till\'e, 2006]{til:06}
Till\'e, Y. (2006).
\newblock {\em Sampling Algorithms}.
\newblock Springer, New York.

\bibitem[Vitter, 1984]{vit:84}
Vitter, J.~S. (1984).
\newblock Faster methods for random sampling.
\newblock {\em Communications of the ACM}, 27(7):703--718.

\bibitem[Vitter, 1985]{vit:85}
Vitter, J.~S. (1985).
\newblock Random sampling with a reservoir.
\newblock {\em {ACM} Transactions on Mathematical Software}, 11(1):37--57.

\bibitem[Vitter, 1987]{vit:87}
Vitter, J.~S. (1987).
\newblock An efficient algorithm for sequential random sampling.
\newblock {\em {ACM} Transactions on Mathematical Software}, 13(1):58--67.

\bibitem[Wilhelm et~al., 2017]{wilh:qual:till:2017}
Wilhelm, M., Qualit\'e, L., and Till\'e, Y. (2017).
\newblock Quasi-systematic sampling from a continuous population.
\newblock {\em Computational Statistics \& Data Analysis}, 105:11--23.

\bibitem[Yates and Grundy, 1953]{yat:gru:53}
Yates, F. and Grundy, P.~M. (1953).
\newblock Selection without replacement from within strata with probability
  proportional to size.
\newblock {\em Journal of the Royal Statistical Society: Series B (Statistical
  Methodology)}, 15:235--261.

\end{thebibliography}
\end{document}